\documentclass[11pt]{article}
\usepackage{fullpage} 
\usepackage[margin=1in]{geometry}

\usepackage{pdfpages}
\usepackage{subcaption}
\usepackage[utf8]{inputenc}

\usepackage[english]{babel}

\renewcommand\thesubfigure{\roman{subfigure}} 
\usepackage[intoc,refpage]{nomencl}

\usepackage{xstring} 
\usepackage{xpatch}

\patchcmd{\thenomenclature}
{\leftmargin\labelwidth}
{\leftmargin\labelwidth}
{}{}

\usepackage{microtype}
\usepackage{pgfplots,wrapfig}
\usepackage{enumitem}
\usepackage[framemethod=TikZ]{mdframed}
\usepackage{lmodern}

\usepackage{tikz}
\usetikzlibrary{positioning}
\usepackage{xcolor}
\definecolor {processblue}{cmyk}{0.96,0,0,0}
\usepackage{setspace}
\usepackage{pgfplots}

\usetikzlibrary{shapes,snakes,shadows,chains}
\usepackage{amsfonts} 
\usepackage{amsmath} 
\usepackage{amsthm}

\usepackage{thmtools}
\usepackage{thm-restate}

\usepackage[colorlinks = True,	linkcolor= blue,citecolor={blue}, linkbordercolor = white]{hyperref}

\usepackage{algorithm}
\usepackage[noend]{algpseudocode}
\usepackage{mathtools}

\DeclarePairedDelimiter{\ceil}{\lceil}{\rceil}

\newtheorem{theorem}{Theorem}
\newtheorem{lemma}{Lemma}

\newtheorem{con}{Conjecture}

\newtheorem{obs}{Observation}[lemma]
\newtheorem{pro}{Proposition}[section]

\newtheorem{prop}{Proposition}[lemma]
\theoremstyle{remark}

\theoremstyle{definition}

\makeatletter
\def\namedlabel#1#2{\begingroup
	#2%
	\def\@currentlabel{\textbf{#2}}%
	\phantomsection\label{#1}\endgroup
}
\makeatother

\usepackage{ifthen}
\newcommand{\fivecyclethree}
{\begin{tikzpicture}[
	vertex/.style={circle,fill=blue!15,draw,minimum size=30pt,inner sep=0pt}]

	\def \n {5}
	\def \radius {3cm}
	\def \margin {8} 
	
	\foreach \s in {1,...,\n}
	{
		
		\node[draw, circle,vertex] at ({360/\n * (\s - 1)}:\radius) {\ifthenelse{\s=2 \OR \s=1}
			{$1$}
			{\ifthenelse{\s=4 }
				{$1,2$}
				{}}};
		\draw[-, >=latex] ({360/\n * (\s - 1)+\margin}:\radius) 
		arc ({360/\n * (\s - 1)+\margin}:{360/\n * (\s)-\margin}:\radius);
	}

	\end{tikzpicture}
	
}

\newcommand{\fivecycletwo}
{\begin{tikzpicture}[
	vertex/.style={circle,fill=blue!15,draw,minimum size=30pt,inner sep=0pt}]

	\def \n {5}
	\def \radius {3cm}
	\def \margin {8} 
	
	\foreach \s in {1,...,\n}
	{
		
		\node[draw, circle,vertex] at ({360/\n * (\s - 1)}:\radius) {\ifthenelse{\s=2 \OR \s=4}
			{$1,2$}
			{}};
		\draw[-, >=latex] ({360/\n * (\s - 1)+\margin}:\radius) 
		arc ({360/\n * (\s - 1)+\margin}:{360/\n * (\s)-\margin}:\radius);
	}

	\end{tikzpicture}
	
}

\newcommand{\sixcyclefour}
{\begin{tikzpicture}[
	vertex/.style={circle,fill=blue!15,draw,minimum size=30pt,inner sep=0pt}]

	\def \n {6}
	\def \radius {3cm}
	\def \margin {8} 
	
	\foreach \s in {1,...,\n}
	{
		
		\node[draw, circle,vertex] at ({360/\n * (\s - 1)}:\radius) {\ifthenelse{\s=3 \OR \s=2 \OR \s=1}
			{$1$}
			{\ifthenelse{\s=5 }
				{$1,2$}
				{}}};
		\draw[-, >=latex] ({360/\n * (\s - 1)+\margin}:\radius) 
		arc ({360/\n * (\s - 1)+\margin}:{360/\n * (\s)-\margin}:\radius);
	}

	\end{tikzpicture}
	
}

\newcommand{\sixcyclethree}
{\begin{tikzpicture}[
	vertex/.style={circle,fill=blue!15,draw,minimum size=30pt,inner sep=0pt}]

	\def \n {6}
	\def \radius {3cm}
	\def \margin {8} 
	
	\foreach \s in {1,...,\n}
	{
		
		\node[draw, circle,vertex] at ({360/\n * (\s - 1)}:\radius) {\ifthenelse{\s=2 \OR \s=4 \OR \s=6}
			{$1,2$}
			{}};
		\draw[-, >=latex] ({360/\n * (\s - 1)+\margin}:\radius) 
		arc ({360/\n * (\s - 1)+\margin}:{360/\n * (\s)-\margin}:\radius);
	}

	\end{tikzpicture}
	
}

\newcommand{\sixcyclethreeb}
{\begin{tikzpicture}[
	vertex/.style={circle,fill=blue!15,draw,minimum size=30pt,inner sep=0pt}]

	\def \n {6}
	\def \radius {3cm}
	\def \margin {8} 
	
	\foreach \s in {1,...,\n}
	{
		
		\node[draw, circle,vertex] at ({360/\n * (\s - 1)}:\radius) {\ifthenelse{\s=2 \OR \s=4 \OR \s=6}
			{$1,2,3$}
			{}};
		\draw[-, >=latex] ({360/\n * (\s - 1)+\margin}:\radius) 
		arc ({360/\n * (\s - 1)+\margin}:{360/\n * (\s)-\margin}:\radius);
	}

	\end{tikzpicture}
	
}

\newcommand{\sixcyclefourb}
{\begin{tikzpicture}[
	vertex/.style={circle,fill=blue!15,draw,minimum size=30pt,inner sep=0pt}]

	\def \n {6}
	\def \radius {3cm}
	\def \margin {8} 
	
	\foreach \s in {1,...,\n}
	{
		
		\node[draw, circle,vertex] at ({360/\n * (\s - 1)}:\radius) 
		{\ifthenelse{\s=3  \OR \s=1 \OR \s=2}{$1,3$}
			{\ifthenelse{\s=5 }	{$1,2,3$}{}}};
		\draw[-, >=latex] ({360/\n * (\s - 1)+\margin}:\radius) 
		arc ({360/\n * (\s - 1)+\margin}:{360/\n * (\s)-\margin}:\radius);
	}

	\end{tikzpicture}

}

\newcommand{\fivecycletwob}
{\begin{tikzpicture}[
	vertex/.style={circle,fill=blue!15,draw,minimum size=30pt,inner sep=0pt}]

	\def \n {5}
	\def \radius {3cm}
	\def \margin {8} 
	
	\foreach \s in {1,...,\n}
	{
		
		\node[draw, circle,vertex] at ({360/\n * (\s - 1)}:\radius) {\ifthenelse{\s=2 \OR \s=4}
			{$1,2,3$}
			{}};
		\draw[-, >=latex] ({360/\n * (\s - 1)+\margin}:\radius) 
		arc ({360/\n * (\s - 1)+\margin}:{360/\n * (\s)-\margin}:\radius);
	}

	\end{tikzpicture}
	
}

\newcommand{\fivecyclethreeb}
{\begin{tikzpicture}[
	vertex/.style={circle,fill=blue!15,draw,minimum size=30pt,inner sep=0pt}]

	\def \n {5}
	\def \radius {3cm}
	\def \margin {8} 
	
	\foreach \s in {1,...,\n}
	{
		
		\node[draw, circle,vertex] at ({360/\n * (\s - 1)}:\radius) {\ifthenelse{\s=2\OR \s=1}
			{$1,3$}
			{\ifthenelse{\s=4 }
				{$1,2,3$}
				{}}};
		\draw[-, >=latex] ({360/\n * (\s - 1)+\margin}:\radius) 
		arc ({360/\n * (\s - 1)+\margin}:{360/\n * (\s)-\margin}:\radius);
	}

	\end{tikzpicture}
	
}

\newcommand\addfiguretwo{
	\setlength{\intextsep}{0pt}%
	\setlength{\columnsep}{18pt}%
	\begin{wrapfigure}{r}{0.5\textwidth}
		\centering
		\resizebox{0.9\linewidth}{!}{
			\begin{tikzpicture}[-latex,node distance =2.4 cm,on grid,
			vertex/.style ={ circle,top color = white, bottom color = processblue!20,draw,processblue, text=black, minimum size=0.8cm},
			state/.style ={ circle,top color = white, bottom color = red!20,draw, red, text=black }, edge/.style = {color=black,-}]
			\node[vertex] [star](0){} ;
			
			\node[vertex][rectangle, right=of 0] (1){\small $v_2^2$};
			
			\node[vertex][circle, right=1.2cm of 1] (2){\small $v_3^2$};
			\node[vertex][star, right=of 2] (3){};
			
			\node[vertex][circle, below=1.2cm of 2] (b){\small $v_3^3$};
			\node[vertex][star, left = of b](a){} ;
			\node[vertex][rectangle, right=1.2cm of b] (c){\small $v_2^3$};
			\node[vertex][star, right=of c] (d){};
			
			\node[vertex][star, above =1.1cm of 1] (A){};
			\node[vertex][star,  left=of A] (B){};
			
			\node[vertex][star, below =1.2cm of c] (G){};
			\node[vertex][star,  right=of G] (F){};
			
			\draw[edge] (1) to (2);
			
			\draw[edge] [red](2) to (b);
			
			\draw[edge] [red](1) to (A);
			\draw[edge] [red](c) to (G);
			
			\draw[edge] (b) to (c);
			\path (0) to node {\dots $P_2^{'}$\dots} (1);
			\path (2) to node {\dots $P_2^{''}$\dots} (3);
			
			\path (a) to node {\dots $P_3^{''}$\dots} (b);
			\path (c) to node {\dots $P_3^{'}$\dots} (d);
			
			\path (A) to node {\dots $P_1$\dots} (B);
			\path (G) to node {\dots $P_4$\dots} (F);
			
			\end{tikzpicture}
		}
		\caption{replacing four paths $P_1, P_2, P_3, P_4$ by three paths}\label{example:2}
		
	\end{wrapfigure}

}

\newcommand\addfigure{
	
	\setlength{\intextsep}{-4pt}%
	\setlength{\columnsep}{18pt}%
	\begin{wrapfigure}{r}{0.5\textwidth}
		\centering
		\resizebox{0.75\linewidth}{!}{
			
			\begin{tikzpicture}[-latex,scale=0.3, auto,swap,node distance =2.2 cm,on grid]
			\tikzstyle{vertex} =[ circle,top color = white, bottom color = processblue!20,draw,black, text=black, minimum size=0.8cm]
			\tikzstyle{edge} =[color=black,-]
			\node[vertex][star] (0){} ;
			\node[vertex][rectangle, right=of 0] (1){\small $u$};
			\node[vertex][rectangle, right=1.3cm of 1] (2){\small $v$};
			\node[vertex][star, right=of 2] (3){};
			
			\node[vertex][star,bottom color =green,black, below=1.3cm of 0] (a){};
			\node[vertex][star,bottom color =green,black, right=of a] (b){};
			\node[vertex][star,bottom color =purple,black, above=1.3cm of 2] (c){};
			\node[vertex][star,bottom color =purple,black, right=of c] (d){ };
			
			\draw[edge] (1) to (2);

			\path (0) to node {\small\dots $P^{'}$\dots} (1);
			\path (2) to node {\small\dots $P^{''}$\dots} (3);
			
			\path (a) to node {\small\dots $P_u$\dots} (b);
			\path (c) to node {\small\dots $P_v$\dots} (d);

			\draw[edge] [red] (1) to (b);
			\draw[edge] [red] (2) to (c);

			\end{tikzpicture}
		}
		\caption{\small replacing three paths $P,P_u,P_v$ by two paths}\label{example:1}
	\end{wrapfigure}
}

\newcommand\crossing{

\begin {tikzpicture}[-latex,node distance =2.4 cm,on grid,
vertex/.style ={ circle,top color = white, bottom color = processblue!20,draw,processblue, text=black, minimum size=0.8cm},
state/.style ={ circle,top color = white, bottom color = red!20,draw, red, text=black }, edge/.style = {color=black,-}]
{
	\node[vertex][star] (0){\small $o_1$} ;
	
	\node[vertex][rectangle, right=of 0] (1){\small $x_1$};

	\node[vertex][rectangle, right=1.2cm of 1] (2){\small $x_2$};

	\node[vertex][star, right=of 2] (3){\small $o_2$};
	
	\draw[edge] (1) to (2);

	\draw[edge] [blue,bend left = 50] (1) to (3);
	\draw[edge] [red,bend left = 50] (0) to (2);
	
	\path (0) to node {\dots $P_1$\dots} (1);
	\path (2) to node {\dots $P_2$\dots} (3);
	
}
\end{tikzpicture}

}

\newcommand\inacceptor{

\begin {tikzpicture}[-latex,node distance =2.4 cm,on grid,
vertex/.style ={ circle,top color = white, bottom color = processblue!20,draw,processblue, text=black, minimum size=0.8cm},
state/.style ={ circle,top color = white, bottom color = red!20,draw, red, text=black }, edge/.style = {color=black,-}]
{
\node[vertex][star] (0){\small $o_1$} ;

\node[vertex][rectangle, right=of 0] (1){\small $x_1$};

\node[vertex][rectangle, right=1.2cm of 1] (2){\small $x_2$};

\node[vertex][circle, below=1.2cm of 2] (a){\small $P_{x_2}$};

\node[vertex][star, right=of 2] (3){\small $o_2$};

\draw[edge] (1) to (2);
\draw[edge][blue] (2) to (a);

\draw[edge] [red,bend left = 50] (1) to (3);

\path (0) to node {\dots $P_1$\dots} (1);
\path (2) to node {\dots $P_2$\dots} (3);

}
\end{tikzpicture}

}

\newcommand\knonouter{

\begin {tikzpicture}[-latex,node distance =2.5 cm,on grid,
vertex/.style ={ circle,top color = white, bottom color = processblue!20,draw,processblue, text=black, minimum size=1cm},
state/.style ={ circle,top color = white, bottom color = red!20,draw, red, text=black }, edge/.style = {color=black,-}]
{
\node[vertex][star] (0){\Large $o_1$} ;

\node[vertex][rectangle, right=of 0] (1){\Large $x_1$};
\node[vertex][ rectangle,right=1.3 cm of 1] (2){\Large $x_2$};
\node[vertex][rectangle, right=of 2] (3){\Large $x_i$};

\node[vertex][rectangle, right=of 3] (k){\Large $x_{k-1}$};
\node[vertex][rectangle, right=1.3cm of k] (4){\Large $x_k$};

\node[vertex][ star, right=of 4] (5){\Large $o_2$};
\draw[edge] (1) to (2);
\draw[edge] (k) to (4);

\draw[edge][bend left = 50,red] (0) to (1);
\draw[edge][bend left = 50,red] (0) to (3);
\draw[edge][bend left = 50,red] (0) to (2);
\draw[edge] [bend left = 50,red] (3) to (5);
\draw[edge] [bend left = 50,red] (4) to (5);
\draw[edge] [bend left = 50,red] (k) to (5);

\path (0) to node {  \dots $P_1$\dots} (1);
\path (2) to node {  \dots } (3);
\path (3) to node {  \dots} (k);

\path (4) to node { \dots $P_2$\dots} (5);
}
\end{tikzpicture}

}

\newcommand\koutercycles{

\begin {tikzpicture}[-latex,node distance =2.5 cm,on grid,
vertex/.style ={ circle,top color = white, bottom color = processblue!20,draw,processblue, text=black, minimum size=1cm},
state/.style ={ circle,top color = white, bottom color = red!20,draw, red, text=black }, edge/.style = {color=black,-}]
{
\node[vertex][star] (0){\Large $o_1$} ;

\node[vertex][rectangle, right=of 0] (1){\Large $x_1$};
\node[vertex][ rectangle,right=1.3cm of 1] (2){\Large $x_2$};
\node[vertex][rectangle, right=of 2] (3){\Large $x_i$};
\node[vertex][rectangle, right=1.3cm of 3] (31){\Large $x_{i+1}$};
\node[vertex][rectangle, right=of 31] (32){\Large $x_{j-1}$};
\node[vertex][rectangle, right=1.3cm of 32] (33){\Large $x_j$};

\node[vertex][rectangle, right=of 33] (k){\Large $x_{k-1}$};
\node[vertex][rectangle, right=1.3cm of k] (4){\Large $x_k$};

\node[vertex][ star, right=of 4] (5){\Large $o_2$};
\node[vertex][circle, below=2cm of 32] (a){\Large $C_1$};

\draw[edge] (1) to (2);
\draw[edge] (k) to (4);
\draw[edge] (32) to (33);
\draw[edge] (3) to (31);

\draw[edge][red] (3) to (a);
\draw[edge][red] (31) to (a);
\draw[edge][red] (32) to (a);
\draw[edge][red] (33) to (a);

\draw[edge][bend left = 50,red] (0) to (1);
\draw[edge][bend left = 50,red] (0) to (3);
\draw[edge][bend left = 50,red] (0) to (2);
\draw[edge] [bend left = 50,red] (33) to (5);
\draw[edge] [bend left = 50,red] (4) to (5);
\draw[edge] [bend left = 50,red] (k) to (5);

\path (0) to node {  \dots $P_1$\dots} (1);
\path (2) to node {  \dots } (3);
\path (31) to node {  \dots} (32);

\path (4) to node { \dots $P_2$\dots} (5);
}
\end{tikzpicture}

}

\newcommand\kouterpath{

\begin {tikzpicture}[-latex,node distance =2.5 cm,on grid,
vertex/.style ={ circle,top color = white, bottom color = processblue!20,draw,processblue, text=black, minimum size=1cm},
state/.style ={ circle,top color = white, bottom color = red!20,draw, red, text=black }, edge/.style = {color=black,-}]
{
\node[vertex][star] (0){$o_1$} ;

\node[vertex][rectangle, right=of 0] (1){\Large $x_1$};
\node[vertex][ rectangle,right=1.3cm of 1] (2){\Large $x_2$};
\node[vertex][rectangle, right=of 2] (3){\Large $x_i$};

\node[vertex][rectangle, right=of 3] (k){\Large $x_{k-1}$};
\node[vertex][rectangle, right=1.3cm of k] (4){\Large $x_k$};

\node[vertex][ star, right=of 4] (5){$o_2$};
\node[vertex][circle, below=2cm of 3] (a){$P_1$};
\draw[edge] (1) to (2);
\draw[edge] (k) to (4);

\draw[edge][red] (3) to (a);
\draw[edge][red]  (1) to (a);
\draw[edge][red]  (2) to (a);
\draw[edge][red]  (4) to (a);
\draw[edge][red]  (k) to (a);

\path (0) to node {  \dots $P_1$\dots} (1);
\path (2) to node {  \dots } (3);
\path (3) to node {  \dots} (k);

\path (4) to node { \dots $P_2$\dots} (5);
}
\end{tikzpicture}

}

\newcommand\ksequence{

\begin {tikzpicture}[-latex,node distance =2.5 cm,on grid,
vertex/.style ={ circle,top color = white, bottom color = processblue!20,draw,processblue, text=black, minimum size=1cm},
state/.style ={ circle,top color = white, bottom color = red!20,draw, red, text=black }, edge/.style = {color=black,-}]
{
\node[vertex][star] (0){$o_1$} ;

\node[vertex][rectangle, right=of 0] (1){\Large $x_1$};
\node[vertex][ rectangle,right=1.2cm of 1] (2){\Large $x_2$};
\node[vertex][rectangle, right=of 2] (3){\Large $x_i$};

\node[vertex][rectangle, right=of 3] (k){\Large $x_{k-1}$};
\node[vertex][rectangle, right=1.2cm of k] (4){\Large $x_k$};

\node[vertex][ star, right=of 4] (5){$o_2$};
\node[vertex][circle, below=2cm of 3] (a){$C_1$};
\node[vertex][circle,right=of a] (b){$C_2$};
\draw[edge] (1) to (2);
\draw[edge] (k) to (4);

\draw[edge][red] (3) to (a);
\draw[edge][red] (3) to (b);

\draw[edge][bend left = 50,red] (0) to (1);
\draw[edge][bend left = 50,red] (0) to (3);
\draw[edge][bend left = 50,red] (0) to (2);
\draw[edge] [bend left = 50,red] (3) to (5);
\draw[edge] [bend left = 50,red] (4) to (5);
\draw[edge] [bend left = 50,red] (k) to (5);

\path (0) to node {  \dots $P_1$\dots} (1);
\path (2) to node {  \dots } (3);
\path (3) to node {  \dots} (k);

\path (4) to node { \dots $P_2$\dots} (5);
}
\end{tikzpicture}

}

\newcommand\splittinginacceptors{

\begin {tikzpicture}[-latex,scale= 0.5, node distance =2.4 cm,on grid,
vertex/.style ={ circle,top color = white, bottom color = processblue!20,draw,processblue, text=black, minimum size=0.8cm},
state/.style ={ circle,top color = white, bottom color = red!20,draw, red, text=black }, edge/.style = {color=black,-}]
{
\node[vertex][star] (0){\small $o_1$} ;

\node[vertex][rectangle, right=of 0] (1){\small $u$};
\node[vertex][rectangle, right=of 1] (2){\small $x_1$};

\node[vertex][rectangle, right=1.2cmof 2] (3){\small $x_2$};

\node[vertex][star, right=of 3] (4){\small $o_2$};

\node[vertex][star,black,bottom color =green, below=1.5cm of 1] (a){\small $o_u$};

\node[vertex][star,black,bottom color =green, right= of a] (b){};
\draw[edge] (2) to (3);

\draw[edge] (1) to (a);

\draw[edge] [bend left = 50,red] (3) to (4);
\draw[edge][bend left = 50,red]  (0) to (2);

\path (0) to node {\dots $P_1^{'}$\dots} (1);
\path (1) to node {\dots $P_1^{''}$\dots} (2);
\path (3) to node {\dots $P_2$\dots} (4);

\path (a) to node {\dots $P_u$\dots} (b);

}
\end{tikzpicture}

}

\newcommand\dangerouslemma{

\begin {tikzpicture}[-latex,node distance =2.4 cm,on grid,
vertex/.style ={ circle,top color = white, bottom color = processblue!20,draw,processblue, text=black, minimum size=0.8cm},
state/.style ={ circle,top color = white, bottom color = red!20,draw, red, text=black }, edge/.style = {color=black,-}]
{
\node[vertex][star] (0){$o_1$} ;

\node[vertex][rectangle, right=of 0] (1){\small $x_1$};
\node[vertex][ circle,right=1.2cm of 1] (2){\small $v_3$};
\node[vertex][rectangle, right=1.2cm of 2] (3){\small $y_1$};
\node[vertex][circle, right=of 3] (4){\small $u$};
\node[vertex][rectangle, right=1.2cm of 4] (5){\small $x_2$};
\node[vertex][ star, right=of 5] (6){$o_2$};

\draw[edge] (1) to (2);
\draw[edge] (2) to (3);
\draw[edge] (4) to (5);

\draw[edge][blue, bend left = 40] (2) to (4);

\draw[edge][red,bend right = 40] (3) to (6);
\draw[edge][red,bend right = 40] (5) to (6);

\path (0) to node {\dots $P_1$\dots} (1);
\path (3) to node {\dots $P_2$\dots} (4);

\path (5) to node {\dots $P_3$\dots} (6);
}
\end{tikzpicture}

}

\newcommand\dangerouslemmatwo{

\begin {tikzpicture}[-latex,node distance =2.4 cm, on grid,
vertex/.style ={ circle,top color = white, bottom color = processblue!20,draw,processblue, text=black, minimum size=0.8cm},
state/.style ={ circle,top color = white, bottom color = red!20,draw, red, text=black }, edge/.style = {color=black,-}]
{
\node[vertex][star] (0){$o_1$} ;

\node[vertex][rectangle, right=of 0] (1){\small $x_2$};
\node[vertex][ circle,right=1.2cm of 1] (2){\small $u$};
\node[vertex][rectangle, right= of 2] (3){\small $y_1$};
\node[vertex][circle, right=1.2cm of 3] (4){\small $v_3$};
\node[vertex][rectangle, right=1.2cm of 4] (5){\small $x_1$};
\node[vertex][ star, right=of 5] (6){$o_2$};

\draw[edge] (1) to (2);
\draw[edge] (3) to (4);
\draw[edge] (4) to (5);

\draw[edge][blue, bend left = 40] (2) to (4);

\draw[edge][red,bend left = 40] (3) to (0);
\draw[edge][red,bend left = 40] (1) to (0);

\path (0) to node {\dots $P_1$\dots} (1);
\path (2) to node {\dots $P_2$\dots} (3);

\path (5) to node {\dots $P_3$\dots} (6);
}
\end{tikzpicture}

}

\newcommand\blockkindtwo[3]{
\begin {tikzpicture}[-latex,scale=0.3,node distance =2.4 cm,on grid,
vertex/.style ={ circle,top color = white, bottom color = processblue!20,draw,processblue, text=black, minimum size=0.8cm},
ver/.style ={ rectangle,top color = white, bottom color = white,draw,white, text=black, minimum size=0.8cm},
state/.style ={ circle,top color = white, bottom color = red!20,draw, red, text=black }, edge/.style = {color=black,-}]
{
\node[vertex][star] (0){$o_1$} ;

\node[vertex][rectangle, right=of 0] (1){\small $x_i$};
\node[vertex][ circle,right=1.2cm of 1] (2){\small $v_4^a$};
\node[vertex][ circle,right=1.2cm of 2] (7){\small $v_4^b$};

\draw[edge] (1) to (2);
\draw[edge] (2) to (7);

\path (0) to node {\dots $P_1$\dots} (1);

\ifthenelse{\equal{#2}{$+\frac{2}{3}$}}
{	\node[vertex][rectangle, right= of 7] (3){\small $y_1$};
\node[vertex][circle, right=1.2cm of 3] (4){\small $v_3$};
\node[vertex][rectangle, right=1.2cm of 4] (5){\small $x_1$};
\node[vertex][ star, right=of 5] (6){$o_2$};
\draw[edge][blue, bend left = 50] (2) to (4);
\draw[edge][red,bend right = 50] (1) to (0);
\draw[edge] (3) to (4);
\draw[edge] (4) to (5);
\path (7) to node {\dots $P_2$\dots} (3);
\path (5) to node {\dots $P_3$\dots} (6);
\node[ver][ below=1 of 1] (a){#1};
\node[ver][ below=1 of 2] (b){#2};
\node[ver][ below=1 of 7] (c){#3};
}
{
\node[vertex][ star, right=of 7] (3){$o_2$};
\path (7) to node {\dots $P_2$\dots} (3);
\node[ver][ below=1 of 1] (a){#1};
\node[ver][ below=1 of 2] (b){#2};
\node[ver][ below=1 of 7] (c){#3};

}

}
\end{tikzpicture}

}

\newcommand\blockkindthree[1]{
\begin {tikzpicture}[-latex,node distance =1.8 cm,on grid,
vertex/.style ={ circle,top color = white, bottom color = processblue!20,draw,processblue, text=black, minimum size=0.5cm},
ver/.style ={ rectangle,top color = white, bottom color = white,draw,white, text=black, minimum size=0.8cm},
state/.style ={ circle,top color = white, bottom color = red!20,draw, red, text=black }, edge/.style = {color=black,-}]
{

\node[vertex][rectangle] (4){\small $X_i$};

\ifthenelse{\equal{#1}{$1$}}
{ 	
\node[vertex][ circle,right= of 4] (A){\small $v_4^a$};
\node[vertex][ circle,right=1.2cm of A] (B){\small $v_4^b$};
\draw[edge] (A) to (B);
\node[ver][ below=1 of A] (b){$+\frac{2}{3}$};
\node[ver][ below=1 of B] (c){$+\frac{2}{3}$};
}
{\node[vertex][ circle,right= of 4] (A){\small $v_3$};
\node[ver][ below=1 of A] (b){$+1$};
}

\node[ver][ below=1 of 4] (a){$-\frac{2}{3}$};

\draw[edge] (4) to (A);

}
\end{tikzpicture}

}

\newcommand\blockkindoneheavy[1]{
\begin {tikzpicture}[-latex,node distance =1.8 cm,on grid,
vertex/.style ={ circle,top color = white, bottom color = processblue!20,draw,processblue, text=black, minimum size=1cm},
ver/.style ={ rectangle,top color = white, bottom color = white,draw,white, text=black, minimum size=0.8cm},
state/.style ={ circle,top color = white, bottom color = red!20,draw, red, text=black }, edge/.style = {color=black,-}]
{

\node[vertex][rectangle] (4){\small $x_i$};

\node[vertex][ circle,right= of 4] (A){\small $v_i$};

\node[ver][ below=1 of 4] (a){$-\frac{5}{3}$};

\draw[edge] (4) to (A);

\IfStrEqCase{#1}{%
{a}{ \node[ver][ below=1 of A] (b){$+\frac{5}{3}$};

}%
{b}{\node[vertex][ rectangle,right= of A] (B){\small $x_{i+1}$};
\node[ver][ below=1 of A] (b){$+\frac{4}{3}$};
\node[ver][ below=1 of B] (b){$-\frac{2}{3}$};
\draw[edge] (B) to (A);
}%
{c}{\node[vertex][ rectangle,right= of A] (B){\small $x_{i+1}$};
\node[ver][ below=1 of A] (b){$+1$};
\node[ver][ below=1 of B] (b){$-\frac{1}{3}$};
\draw[edge] (B) to (A);}	%

{d}{\node[vertex][ rectangle,right= of A] (B){\small $X_{i+1}$};
\node[ver][ below=1 of A] (b){$+1$};
\node[ver][ below=1 of B] (b){$-\frac{1}{3}$};
\draw[edge] (B) to (A);}	%

}%

}
\end{tikzpicture}

}

\newcommand\blockkindoneend{
\begin {tikzpicture}[-latex,node distance =1.8 cm,on grid,
vertex/.style ={ circle,top color = white, bottom color = processblue!20,draw,processblue, text=black, minimum size=1cm},
ver/.style ={ rectangle,top color = white, bottom color = white,draw,white, text=black, minimum size=0.8cm},
state/.style ={ circle,top color = white, bottom color = red!20,draw, red, text=black }, edge/.style = {color=black,-}]
{

\node[vertex][rectangle] (A){\small $B_i$};
\node[vertex][ rectangle,right= of A] (B){\small $X_{i+1}$};

\node[ver][ below=1 of A] (a){$-\frac{2}{3}$};
\node[ver][ below=1 of B] (b){$-\frac{1}{3}$};

\draw[edge] (B) to (A);

}
\end{tikzpicture}

}

\newcommand\blockkindone[2]{
\begin {tikzpicture}[-latex,node distance =1.8 cm,on grid,
vertex/.style ={ circle,top color = white, bottom color = processblue!20,draw,processblue, text=black, minimum size=0.5cm},
ver/.style ={ rectangle,top color = white, bottom color = white,draw,white, text=black, minimum size=0.8cm},
state/.style ={ circle,top color = white, bottom color = red!20,draw, red, text=black }, edge/.style = {color=black,-}]
{

\node[vertex][rectangle] (4){\small $x_i$};

\node[vertex][ circle,right= of 4] (A){\small $v_i$};
\node[ver][ below=1 of A] (b){#2};

\node[ver][ below=1 of 4] (a){#1};

\draw[edge] (4) to (A);

}
\end{tikzpicture}

}

\newcommand\oppositesides{

\begin {tikzpicture}[-latex,node distance =2.4 cm,on grid,
vertex/.style ={ circle,top color = white, bottom color = processblue!20,draw,processblue, text=black, minimum size=0.8cm},
state/.style ={ circle,top color = white, bottom color = red!20,draw, red, text=black }, edge/.style = {color=black,-}]
{
\node[vertex][star] (0){$o_1$} ;

\node[vertex][rectangle, right=of 0] (1){\small $x_1$};
\node[vertex][ circle,right=1.2cm of 1] (2){\small $v_3$};
\node[vertex][rectangle, right=1.2cm of 2] (3){\small $y_1$};
\node[vertex][circle, right=of 3] (4){\small $u$};
\node[vertex][rectangle, right=1.2cm of 4] (5){\small $x_2$};
\node[vertex][ star, right=of 5] (6){$o_2$};

\draw[edge] (1) to (2);
\draw[edge] (2) to (3);
\draw[edge] (4) to (5);

\draw[edge][bend left = 60] (2) to (4);
\path (0) to node {\dots $P_1$\dots} (1);
\path (3) to node {\dots $P_2$\dots} (4);

\path (5) to node {\dots $P_3$\dots} (6);
}
\end{tikzpicture}

}

\newcommand\caseexternalb{

\begin {tikzpicture}[-latex,node distance =2 cm,on grid,sibling distance=1cm,
vertex/.style ={ circle,top color = white, bottom color = processblue!20,draw,processblue, text=black, minimum size=1.3cm},
state/.style ={ circle,top color = white, bottom color = red!20,draw, red, text=black }, edge/.style = {color=black,-}]

\node[vertex][star](0){\Large $o_1$} ;
\node[vertex][circle, right=2cm of 0] (1){};
\node[vertex][circle, right=of 1] (2){};
\node[vertex][rectangle, right=2cm of 2] (3){\Large $x_1$};

\node[vertex] [star,bottom color =yellow,black, above=2cm of 3] (P){\Large $o_{x_1}$} ;
\node[vertex][star,bottom color =yellow,black, left=3cm of P] (Q){};

\draw[edge] (0) to (1);
\draw[edge] (2) to (3);
\path (1) to node {\dots } (2);
\tikzset{vertex/.append style = {bottom color =red,black}};

\node[vertex][star, below=2cm of 0] (a){\Large $o_2$};
\node[vertex][circle, right=2cm of a] (b){};
\node[vertex][circle, right=of b] (c){};
\node[vertex][rectangle, right=2cm of c] (d){\Large $x_2$};

\node[vertex][circle,bottom color =green,black, below=2cm of d] (A){\Large $C$};

\path (P) to node { \dots$P_{x_1}$ \dots} (Q);

\draw[edge] (d) to (A);

\draw[edge] (a) to (b);
\draw[edge] (c) to (d);
\path (b) to node {\dots } (c);

\draw[edge] (3) to (P);

\end{tikzpicture}
}

\newcommand\caseexternalc{

\begin {tikzpicture}[-latex,node distance =2 cm,on grid,sibling distance=1cm,
vertex/.style ={ circle,top color = white, bottom color = processblue!20,draw,processblue, text=black, minimum size=1.3cm},
state/.style ={ circle,top color = white, bottom color = red!20,draw, red, text=black }, edge/.style = {color=black,-}]

\node[vertex][star](0){\Large $o_1$} ;
\node[vertex][circle, right=2cm of 0] (1){};
\node[vertex][circle, right=of 1] (2){};
\node[vertex][rectangle, right=2cm of 2] (3){\Large $x_1$};

\node[vertex] [star,bottom color =yellow,black, above=2cm of 3] (P){\Large $o_{x_1}$} ;
\node[vertex][star,bottom color =yellow,black, left=3cm of P] (Q){};

\draw[edge] (0) to (1);
\draw[edge] (2) to (3);
\path (1) to node {\dots } (2);
\tikzset{vertex/.append style = {bottom color =red,black}};

\node[vertex][star, below=2cm of 0] (a){\Large $o_2$};
\node[vertex][circle, right=2cm of a] (b){};
\node[vertex][circle, right=of b] (c){};
\node[vertex][rectangle, right=2cm of c] (d){ \Large$x_2$};

\node[vertex][star,bottom color =green,black, below=2cm of d] (A){ \Large $o_{x_2}$};

\node[vertex][star,bottom color =green,black, left=3cm of A] (B){};

\path (A) to node { \dots$P_{x_2}$ \dots} (B);
\path (P) to node { \dots$P_{x_1}$ \dots} (Q);

\draw[edge] (d) to (A);

\draw[edge] (a) to (b);
\draw[edge] (c) to (d);
\path (b) to node {\dots } (c);

\draw[edge] [color=blue] (3) to (a);
\draw[edge] [color=red] (3) to (P);

\draw[edge][color=green] (3) to (B);

\draw[edge][bend left = 30,color=blue] (0) to (3);
\end{tikzpicture}
}

\newcommand\caseexternala{

\begin {tikzpicture}[-latex,node distance =2 cm,on grid,sibling distance=1cm,
vertex/.style ={ circle,top color = white, bottom color = processblue!20,draw,processblue, text=black, minimum size=1.3cm},
state/.style ={ circle,top color = white, bottom color = red!20,draw, red, text=black }, edge/.style = {color=black,-}]

\node[vertex][star](0){\Large $o_2$} ;
\node[vertex][circle, right=2cm of 0] (1){};
\node[vertex][circle, right=of 1] (2){};
\node[vertex][rectangle, right=2cm of 2] (3){\Large $x_2$};

\draw[edge] (0) to (1);
\draw[edge] (2) to (3);
\path (1) to node {\dots } (2);
\tikzset{vertex/.append style = {bottom color =red,black}};

\node[vertex][star, below=2cm of 0] (a){\Large $o_1$};
\node[vertex][circle, right=2cm of a] (b){};
\node[vertex][circle, right=of b] (c){};
\node[vertex][rectangle, right=2cm of c] (d){\Large $x_1$};

\node[vertex][star,bottom color =green,black, below=2cm of d] (A){\Large $o_{x_1}$};

\node[vertex][star,bottom color =green,black, left=3cm of A] (B){};

\path (A) to node { \dots$P_{x_1}$ \dots} (B);

\draw[edge] (d) to (A);

\draw[edge] (a) to (b);
\draw[edge] (c) to (d);
\path (b) to node {\dots } (c);

\draw[edge] [color=green] (3) to (a);

\draw[edge][bend left = 40,color=blue] (0) to (3);
\end{tikzpicture}
}

\newcommand\caseexternaltwo{
\begin {tikzpicture}[-latex,node distance =2 cm,on grid,sibling distance=1cm,
vertex/.style ={ circle,top color = white, bottom color = processblue!20,draw,processblue, text=black, minimum size=1.3cm},
state/.style ={ circle,top color = white, bottom color = red!20,draw, red, text=black }, edge/.style = {color=black,-}]

\node[vertex][rectangle] (a){\Large $x_1$};
\node[vertex][circle, right=2cm of a] (b){};
\node[vertex][circle, right=of b] (c){};
\node[vertex][rectangle, right=2cm of c] (d){\Large $x_2$};

\node[vertex] [rectangle,bottom color =yellow,black, above=2cm of a] (P){\Large $o_{x_1}$} ;
\node[vertex][star,bottom color =yellow,black, left=3cm of P] (Q){};

\node[vertex][star,bottom color =green,black, below=2cm of d] (A){\Large $o_{x_2}$};

\path (P) to node {\dots$P_{x_1}$ \dots} (Q);

\draw[edge] (d) to (A);

\draw[edge] (a) to (b);
\draw[edge] (c) to (d);
\path (b) to node {\dots } (c);

\draw[edge] [color=red] (P) to (a);

\end{tikzpicture}

}

\newcommand\between{

\begin {tikzpicture}[-latex,node distance =2.4 cm,on grid,
vertex/.style ={ circle,top color = white, bottom color = processblue!20,draw,processblue, text=black, minimum size=0.8cm},
state/.style ={ circle,top color = white, bottom color = red!20,draw, red, text=black }, edge/.style = {color=black,-}]
{	

\node[vertex][star] (0){} ;

\node[vertex][circle, right=of 0] (1){\small $u$};

\node[vertex][rectangle, right=1.2cm of 1] (2){\small $x_2$};

\node[vertex][rectangle, right=of 2] (3){\small $x_1$};

\node[vertex][circle, right=1.2cm of 3] (4){\small $v_3$};
\node[vertex][star, right=of 4] (5){};

\draw[edge] (1) to (2);
\draw[edge] (3) to (4);
\draw[edge][bend left = 50] (1) to (4);
\path (0) to node {\dots $P_1$\dots} (1);
\path (2) to node {\dots $P_2$\dots} (3);

\path (4) to node {\dots $P_3$\dots} (5);
}
\end{tikzpicture}

}

\newcommand\samesides[4]{

\begin {tikzpicture}[-latex,node distance =2.4 cm,on grid,
vertex/.style ={ circle,top color = white, bottom color = processblue!20,draw,processblue, text=black, minimum size=0.8cm},
state/.style ={ circle,top color = white, bottom color = red!20,draw, red, text=black }, edge/.style = {color=black,-}]
{	
\node[vertex][star] (0){} ;
\node[vertex][rectangle, right=of 0] (1){\small #1};
\node[vertex][circle, right=1.2cm of 1] (2){\small #2};

\node[vertex][rectangle, right=of 2] (3){\small #3};
\node[vertex][circle, right=1.2cm of 3] (4){\small #4};
\node[vertex][star, right=of 4] (5){};

\draw[edge] (1) to (2);
\draw[edge] (3) to (4);
\draw[edge][bend left = 50] (2) to (4);
\path (0) to node {\dots $P_1$\dots} (1);
\path (2) to node {\dots $P_2$\dots} (3);

\path (4) to node {\dots $P_3$\dots} (5);
}
\end{tikzpicture}

}

\newcommand\difpath[4]{

\begin {tikzpicture}[-latex,scale=0.3,node distance =2.2 cm,on grid,
vertex/.style ={ circle,top color = white, bottom color = processblue!20,draw,processblue, text=black, minimum size=0.8cm},
state/.style ={ circle,top color = white, bottom color = red!20,draw, red, text=black }, edge/.style = {color=black,-}]

\node[vertex][star] (0){} ;
\node[vertex][rectangle, right=of 0] (1){\small #1};
\node[vertex][circle, right=1.2cm of 1] (2){\small #2};
\node[vertex][star, right=of 2] (3){};

\node[vertex][circle, below=1.2cm of 2] (b){\small #3};
\node[vertex] [star, left =of b](a){} ;

\node[vertex][rectangle, right=1.2cm of b] (c){\small #4};

\node[vertex][star, right=of c] (d){};

\draw[edge] (1) to (2);

\draw[edge] (2) to (b);

\draw[edge] (b) to (c);
\path (0) to node {\dots $P_1^{'}$\dots} (1);
\path (2) to node {\dots $P_1^{''}$\dots} (3);

\path (a) to node {\dots $P_2^{''}$\dots} (b);
\path (c) to node {\dots $P_2^{'}$\dots} (d);

\end{tikzpicture}
}

\newcommand\difshiled[4]{
\vspace*{5mm}	
\begin {tikzpicture}[-latex,scale=0.3,node distance =2.2 cm,on grid,
vertex/.style ={ circle,top color = white, bottom color = processblue!20,draw,processblue, text=black, minimum size=0.8cm},
state/.style ={ circle,top color = white, bottom color = red!20,draw, red, text=black }, edge/.style = {color=black,-}]

\tikzset{vertex/.append style = {star}};
\node[vertex](A){};

\tikzset{vertex/.append style = {circle}};
\node[vertex][ right=of A] (B){\small #3};
\node[vertex][ right=1.2cm of B] (C){\small #4};
\node[vertex][star, right=of C] (D){};

\draw[edge] (C) to (B);
\path (A) to node {\dots $P_2^{''}$\dots } (B);
\path (D) to node {\dots $P_1^{''}$\dots } (C);

\tikzset{vertex/.append style = {bottom color =green,black}}

\node[vertex][ star,below=1.2cm  of A] (0){};
\node[vertex][ rectangle,below=1.2cm of D] (2){\small #1};

\node[vertex][circle, left=1.2cm of 2] (1){};

\draw[edge] (1) to (2);
\path (0) to node {$- - - P_1^{'}- - -$  } (1);
\tikzset{vertex/.append style = {bottom color =red,black}};

\node[vertex][star, below=1.2cm of 0] (a){};

\node[vertex][circle, below=1.2cm  of 1] (b){};
\node[vertex][rectangle, right=1.2cm  of b] (c){\small #2};

\draw[edge] (b) to (c);
\path (a) to node { $- - - P_2^{'}- - -$ } (b);
\end{tikzpicture}
}

\title{On the path partition number of 6-regular graphs}

\begin{document}

\author{Uriel Feige \thanks{The Weizmann Institute of Science, Israel. Email: {\tt uriel.feige@weizmann.ac.il}.}\and Ella Fuchs \thanks{The Weizmann Institute of Science, Israel. Email: {\tt ella.fuchs@weizmann.ac.il}.}}
\maketitle
\thispagestyle{empty}

\begin{abstract}
	
	A \textit{path partition} (also referred to as a linear forest)
	of a graph $G$ is a set of vertex-disjoint paths which together contain all the vertices of $G$. An isolated vertex is considered to be a path in this case. The path partition conjecture states that every $n$-vertices $d$-regular graph has a path partition with at most $\frac{n}{d+1}$ paths. The conjecture has been proved for all $d<6$. We prove the conjecture for $d=6$.
	
\end{abstract}

\section{Introduction}

Let $G(V,E)$ be a simple undirected graph, where the cardinality of $V$ and $E$ are respectively denoted by $n$ and $m$. A {\em path} (also referred to as a simple path) in a graph $G$ is a sequence of $t\geq 1$ distinct vertices, $(v_1, v_2,\dots, v_t)$ such that for every $1\leq i\leq t-1$, $(v_i, v_{i+1})\in E$. These edges $(v_i, v_{i+1})$ are referred to as the path edges. The size of a path is the number of vertices in the sequence. A path of size one has no edges.
A {\em path cover} of $G$ is a set of paths such that every vertex in $V$ belongs to {\em at least} one of the paths, whereas a {\em path partition} of $G$ is a set of paths such that every vertex in $V$ belongs to {\em exactly} one path. (In some of the related literature the term {\em path cover} is used in the sense that we refer to here as {\em path partition}.)
Given a path partition, we refer to each path in the set as a {\em component}. The cardinality of a path partition refers to the number of components in it.
\pgfdeclarelayer{background}
\pgfsetlayers{background,main}

{\em The path partition number} of $G$, denoted by $\pi_{p}(G)$, is the minimum cardinality among all path partitions of $G$. The problem of finding a path partition of minimum cardinality is called {\em the path partition problem}. This problem is NP-hard \cite{Hartmanis82}, since it contains the Hamiltonian path problem as a special case. There are polynomial-time algorithms for the path partition problem for some families of graphs, including (among others)  
forests \cite{skupien1974path}, interval graphs \cite{HUNG2011648}, circular arc graphs, bipartite permutation graphs \cite{srikant1993optimal}, block graphs \cite{pak1999optimal}, and cographs \cite{chang19962}.

An interesting theoretical question is to provide upper bounds on $\pi_{p} (G)$ that apply to large families of graphs. Ore~\cite{ore1961arc} proved that given a graph $G$ with $n$ vertices, if $\pi_{p}(G)\geq 2$ then $\pi_{p}(G) \leq n- \sigma_2(G)$, where $\sigma_2(G)$ is the minimum sum of degrees of two non-adjacent vertices. This theorem implies that a graph has a Hamiltonian path if $\sigma_2(G)\geq n-1$. 
This theorem was generalized by Noorvash have expanded the  \cite{noorvash1975} who found a relation between the number of edges in a graph and $\pi_{p}(G)$. Another early contribution to bound the path partition number is the Gallai-Milgram Theorem  \cite{gallai1960verallgemeinerung}, which states that the independence number $\alpha(G)$ is an upper bound for this number. Namely, $\pi_{p} (G) \leq  \alpha(G)$. This holds not only for undirected graphs, but also for a corresponding notion of path partitions for directed graphs. The minimum degree in $G$, denoted by $\delta(G)$ can also provide lower bounds on the path partition number. By the classical theorem of Dirac \cite{dirac1952some}, if $\delta(G) \geq n/2$ then $G$ has an Hamiltonian cycle.  
It is worth mentioning that there are graphs where the path partition number can be very large. For example, the star graph on $n$ vertices, requires at least $n-2$ components in any path partition.

Our work relates to the following conjecture of Magnant and Martin \cite{MagnantM09}.

\begin{con}{Partition number conjecture:} \label{con1} For every $d$-regular graph $G$, \[ \pi_{p} (G)\leq \frac{n}{d+1}\]
\end{con}

Megnant and Martin proved their conjecture for all $d\leq 5$.
The upper bound in the conjecture is tight for a
graph containing  $\frac{n}{d+1}$ copies of cliques on $d+1$ vertices ($K_{d+1}$).
However, for connected regular graphs, some improvement is possible. For connected cubic graphs, Reed \cite{Reed96} provided a sharper bound $ \pi_{p} (G)\leq \ceil{\frac{n}{9}}$. In the same paper, Reed conjectured that every $2$-connected $3$-regular graph has a path partition with at most $\lceil \frac{n}{10}\rceil$ components. This has been recently confirmed by Yu  \cite{yu2018covering}. For every $d \geq 4$, there are connected graphs (even 2-connected) for which the path partition number of $G$ is at least $\frac{n(d-3)}{d^2+1}$. In particular, for $d\geq 13$, there are connected graphs that require $\frac{n}{d+4}$ components in any path partition.
See examples in \cite{yu2018covering,suil2010balloons,suil2011matching}.
It is worth mentioning that almost all $d$-regular
graphs are Hamiltonian for $d \geq 3$ \cite{robinson1994almost}. For a broader literary review and more information see \cite{manuel2018revisiting}.

\subsection{Related work}

The set of edges in a path partition is also referred to as a {\em linear forest}. Given a graph $G(V,E)$, its {\em linear arboricity}, denoted by $la(G)$, is the minimum number of disjoint linear forests in $G$ whose union is all $E$. This notion was introduced by Harary in \cite{harary1970covering}. The following conjecture, known as the linear arboricity conjecture, was raised in \cite{Akiyama1980}:

\begin{con}\label{con3}
	For every $\Delta$ and every graph $G$ of maximum degree $\Delta$, \[la(G)\leq\ceil*{\frac{\Delta+1}{2}}\]
\end{con}

Every graph of maximum degree $\Delta$ can be embedded in some $d$ regular graph for $d=\Delta$ (this may require adding vertices). For $d$-regular graphs $ la(G)\geq \ceil*{\frac{d+1}{2}}$. This is because each linear forest $G$ can use only two of the edges for any vertex. Hence if $d$ is odd, then there are at least $\frac{d+1}{2} $ linear forests. For the case where $d$ is even, at least one of the vertices is an end-vertex of a path in a linear forest, giving us at least $\frac{d}{2} +1 $ linear forests. Therefore, the linear arboricity conjecture is equivalent to proving that $ la(G)=\ceil*{\frac{d+1}{2}}$ for every $d$-regular graph $G$.

As $G$ has $\frac{nd}{2}$ edges, and each one of them is in at least one linear forest, this implies that given $la(G)$  there is a linear forest with at least $\frac{\frac{nd}{2}}{la(G)}$ edges. Note that the cardinality of a path partition is exactly $n$ minus the number of edges in it. Hence
\begin{equation}\label{eq:linear}
\pi_p(G)\leq n- \frac{\frac{nd}{2}}{la(G)}
\end{equation}

\begin{pro}
	\label{pro:odd}
	For $d$ odd,  \hyperref[con3] {\text{Conjecture} \ref{con3}} implies \hyperref[con1] {\text{Conjecture} \ref{con1}}.
\end{pro}

\begin{proof}
	
	Assuming  \hyperref[con3] {\text{Conjecture} \ref{con3}} to be true, for odd $d$  we have that  $la(G)=\frac{d+1}{2}$. Plugging it to inequality (\ref{eq:linear}) implies $\pi_p(G)\leq n- \frac{\frac{nd}{2}}{\frac{d+1}{2}} =  \frac{n}{d+1}$.
	
\end{proof}

We state a relaxed version of \hyperref[con1] {\text{Conjecture} \ref{con1}}:
\begin{con}{The relaxed partition number conjecture:}
	\label{con2} For any $d$-regular graph $G$,
	\[ \pi_{p} (G)= O\left(\frac{n}{d+1}\right)\]
\end{con}

\begin{pro}
	\label{pro:even}
	For $d$ even, \hyperref[con3] {\text{Conjecture} \ref{con3}} implies \hyperref[con2] {\text{Conjecture} \ref{con2}}.
\end{pro}

\begin{proof}
	Assuming \hyperref[con3] {\text{Conjecture} \ref{con3}} to be true, for even $d$ we have that $la(G)=\frac{d+2}{2}$. Plugging it to inequality (\ref{eq:linear}) implies $\pi_p(G)\leq n- \frac{\frac{nd}{2}}{\frac{d+2}{2}}  = \frac{2n}{d+2}$.
	
\end{proof}

The linear arboricity conjecture has been proved in the special cases of $d = 3,4,5,6,8$ and $10$ \cite{Akiyama1980,akiyama1981covering,enomoto1984linear,guldan1986linear}. (Remark: these results imply \hyperref[con1] {\text{Conjecture} \ref{con1}} for $d = 3,5$, by \hyperref[pro:odd] {\text{Proposition} \ref{pro:odd}}.)
The linear arboricity conjecture was shown to be asymptotically correct as $d \to \infty$. Alon \cite{alon1988linear} showed that for every $d$-regular graph $G$, $la(G) \leq  \frac{d}{2} + O\left(\frac{d\log \log d}{\log d}\right)$. This result was subsequently improved to  $la(G) \leq  \frac{d}{2} + O\left(d^{\frac{2}{3}} (\log d)^{\frac{1}{3}}\right)$ \cite{alon2004probabilistic}, and to $la(G)\leq  \frac{d}{2} + O\left(d ^{\frac{2}{3}-c}\right)$, for some constant $ c>0$ \cite{ferber2019towards}.

Plugging the best result on the linear arboricity conjecture to inequality (\ref{eq:linear}) gives the following asymptotic bounds on the path partition number of regular graphs.

$$\pi_p(G) \leq O\left( \frac{n}{ d^{\frac{1}{3}+c} }\right)  $$

The linear arboricity can be thought of as an integer programming problem, where for each linear forest $F_i$ of $G$ we chose a weight $\alpha_i\in \{0,1\}$, and the goal is to minimize  $\sum_{i} \alpha_i$ under the constraint that for every $e\in E$, $\sum_{i|e\in F_i} \alpha_i \geq 1$. The relaxed version of this problem, where we are allowed to pick a linear forest fractionally ($0\leq \alpha_i\leq 1$) , is called the {\em fractional linear arboricity} and is denoted by $fla(G)$.


Given $fla(G)$ for some $\alpha_i^*$ we get that:
$$|E(G)| \leq \sum_{e\in E} \sum_{i|e\in F_i}\alpha_i^*=\sum_{i}\sum_{e|e\in F_i} \alpha_i^*= \sum_{i} \alpha_i^* |F_i|\leq \max_{i}|F_i|\cdot fla(G)$$

Hence the number of edges in the linear forest of maximal size is at least $\frac{\frac{nd}{2}}{fla(G)}$. Therefore, the following inequality holds:


\begin{equation}\label{eq:fraclinear}
\pi_p(G)\leq n- \frac{\frac{nd}{2}}{fla(G)}
\end{equation}

Feige, Ravi and Singh~\cite{feige2014short} proved that $fla(G)= \frac{d}{2} + O(\sqrt{d})$, and deduced from inequality (\ref{eq:fraclinear}) that $$\pi_p(G) \leq O\left( \frac{n}{ \sqrt{d}}\right)$$
They also proved that if \hyperref[con2] {\text{Conjecture} \ref{con2}} holds, then every $n$-vertex $d$-regular graph has a tour of length $(1 + O(\frac{1}{d}))n$ visiting all its vertices.

\subsection{Our contribution }

The main purpose of the paper is proving \hyperref[con1] {\text{Conjecture} \ref{con1}} for $d=6$.

Our main result is the following theorem:

\begin{mdframed}[hidealllines=true,backgroundcolor=gray!25]
	\vspace{-5pt}
	
	\begin{theorem}	\label{thm:main}
		Every $6-$regular graph with $n$ vertices, has a path partition whose cardinality is at most $\frac{n}{7}$. Equivalently, the average size of a component in this path partition is at least $7$.
	\end{theorem}
\end{mdframed}

In passing, we also show the following theorem for $d=5$.

\begin{mdframed}[hidealllines=true,backgroundcolor=gray!25]
	\vspace{-5pt}
	
	\begin{theorem}	\label{thm:main5}
		Every $5$-regular graph with $n$ vertices and no $K_6$ (clique of size $6$), has a path partition whose cardinality is at most $\frac{n}{6+\frac{1}{3}}$. Equivalently, the average size of a component in this path partition is at least $6+\frac{1}{3}$.
	\end{theorem}
\end{mdframed}

\section{Proof overview}\label{sec:overview}

Given a path partition for a graph $G(V,E)$ we distinguish between three types of components.
\nomenclature[a]{\em Path partition}{A path partition (also referred to as a linear forest)
	of a graph $G$ is a set of paths such that every vertex in $V$ belongs to exactly one path.
}%

\begin{enumerate}\label{components}	
	\item Cycles. A component of size $t \ge 3$ is referred to as a {\em cycle component} if the induced graph on the vertices of the component contains a spanning cycle.
	
	\item Isolated vertices. A component of size $1$.
	
	\item Paths. All the rest of the components are referred to as {\em path components}.


\end{enumerate}\label{can}

Given a 6-regular graph $G(V,E)$, a path partition is said to be {\em canonical} if it satisfies the following properties:

\begin{enumerate}
	
	\item It has the smallest number of components.
	
	\item Conditioned on the first property, it has the largest number of cycles.
	
	\item It has no isolated vertices.
	
\end{enumerate}

\hyperref[lem:canonical]{\text{Lemma} \ref{lem:canonical}} (extension of the work of~\cite{MagnantM09}, see Section~\ref{sec:canonical}) shows that given a partition satisfying the first two properties, one can extend it to satisfy the third property as well. Thus, every $6$-regular graph has a canonical path partition as defined above. We wish to show that the average size of a component in a canonical path partition is at least~7. We first introduce some notation that will assist in explaining the main ideas in our proof.

Given a canonical path partition, we partition the set of edges $E$ into three sets:

\begin{enumerate}
	
	\item Path edges $E_P$: this set includes all edges of the path partition that belong to components that are paths.
	
	\item Cycle edges $E_C$: this set includes all edges (whether part of the path partition or not) both whose endpoints are in a cycle component. Observe that in a canonical path partition there is no edge that has its endpoints in different cycles, as then the two cycles can be replaced by one component.
	
	\item Free edges $E_F$. These are all the remaining edges. None of them is part of the path partition, and the endpoints of a free edge either lie in two different components (provided that not both of them are cycles), or within the same path component.
	
\end{enumerate}

We partition the set of vertices $V$ into five disjoint classes. A vertex belongs to the first applicable class:

\nomenclature[e4]{{\em balanced edges}}{A free edge with one vertices in $V_1$ and the other is in $V_2$ (see definition bellow).}%
\nomenclature[v0]{$V_1$}{end-vertices of paths, and vertices of cycles.}%
\nomenclature[v0]{$V_2$}{path vertices that are connected by a free edge to a vertex in $V_1$.}%

\nomenclature[v3]{$V_3$}{path vertices that are connected by two path edges to vertices in $V_2$.}%

\nomenclature[v5]{$V_4$}{path vertices that are connected by exactly one path edge to a vertex in $V_2$.}%
\nomenclature[v6]{$V_5$}{the remaining vertices.}%

\begin{enumerate}

	\item $V_1$: end-vertices of paths, and vertices of cycles.
	Observe that in a canonical path partition, there is no free edge that has its end-vertices both in $V_1$, as then either the two components can be replaced by one component, or a path components can be made into a cycle. Hence the graph $G(V_1,E_F)$ (containing only vertices from $V_1$ and their induced free edges) is an independent set.
	
	
	\item $V_2$: path vertices that are connected by a free edge to a vertex in $V_1$.
	
	\item $V_3$: path vertices that are connected by two path edges to vertices in $V_2$.

	\item $V_4$: path vertices that are connected by exactly one path edge to a vertex in $V_2$.
	
	\item $V_5$: the remaining vertices.
	
\end{enumerate}
%

%

%
%

A free edge is {\em balanced} 
if one of its end-vertices is in $V_1$ and the other is in $V_2$. Observe that all the free edges incident to $V_1$ are balanced.
Balanced edges play a key role in our analysis. We now provide some intuition of how they may be used, and what are the issues that need to be handled. Suppose, for simplicity, that the canonical path partition contains $p$ paths and no cycles. This implies that $|V_1| = 2p$. As every vertex in $V_1$ is incident with exactly~5 balanced edges, we have that the number of balanced edges is $10p$. As every vertex in $V_2$ is incident with at most~4 balanced edges, we have that $|V_2| \ge \frac{5}{2}p$. It follows that $n=|V| \ge |V_1| + |V_2| \ge \frac{9p}{2}$. This falls short of the bound of $n \ge 7p$ that we would like to prove. Here are two alternatives, each of which by itself suffices in order to fill in the missing gap.

\begin{enumerate}
	
	\item Show that on average, every vertex in $V_2$ is incident with only two balanced edges, rather than four.
	
	\item Show that $|V_3 \cup V_4 \cup V_5| \ge \frac{5}{2}p$.
	
\end{enumerate}

Our proof will not show that either of the alternatives holds, but rather will show that a ``convex combination" of these alternatives hold. That is, in every canonical path partition, some of the $V_2$ vertices have fewer than four balanced edges. In addition there are vertices of types $V_3 \cup V_4 \cup V_5$, and the combination of these two aspects has a quantitative effect that is sufficiently large so as to conclude that $n \ge 7p$.

\addfigure A key observation in this proof plan (made also in~\cite{MagnantM09}) is that if two vertices $u,v \in V_2$ are neighbors in a path $P$, then this severely limits the number of balanced edges that they can consume.
For example, it cannot be that $u$ has a balanced edge to an end-vertex of path $P_u$ and $v$ has a balanced edge to an end-vertex of path $P_v\neq P_u$, because then the path partition is not canonical: path $P$ can be eliminated by appending one part of it to $P_u$ and the other part to $P_v$. (See Figure \ref{example:1}.)

Based on the above observation, we may infer (in some approximate sense that is made rigorous in our proof) that in a canonical path partition, a vertex from $V_2$ is either incident with only two balanced edges, or is not a path-neighbor of any other vertex from $V_2$. The former case corresponds to alternative~1, whereas the latter case is related to alternative~2 in the following sense: if no two vertices in $V_2$ are path-neighbors, {\em and in addition all paths have even size}, 
then necessarily $|V_3 \cup V_4 \cup V_5| \ge |V_2| \ge \frac{5}{2}p$, proving alternative~2.

However, the canonical path partition might contain paths of odd size, and then the argument above does not suffice. For example, a path of size~3 may have its middle vertex $v$ belong to $V_2$, the vertex $v$ might be incident with four balanced edges, and yet the path contributes no vertex to $V_3 \cup V_4 \cup V_5$. To compensate for this, we must have other paths in which there are strictly more vertices in $V_3 \cup V_4 \cup V_5$ than in $V_2$. Observe that indeed there must be paths of size greater than~3, because if all paths are of size~3 we have that $|V_2| \le \frac{1}{2}|V_1|$, contradicting the fact that $|V_2| \ge \frac{5}{4}|V_1|$. Observe also that in every path $P$ we have that $|V_3 \cap P| \le |V_2 \cap P|$, and hence for path $P$ to compensate for a path of size~3, the set $(V_4 \cup V_5) \cap P$ must be nonempty.

How can we infer that $V_4 \cup V_5$ has substantial size (where the interpretation of substantial depends on the extent to which alternative~1 fails to hold)? The key is to consider the set $E_F(V_3)$ of free edges incident with vertices from $V_3$. These edges cannot be incident with vertices from $V_1$ (as their free edges are incident with vertices from $V_2$). If edges from $E_F(V_3)$ are incident with vertices from $V_2$, this brings us closer to alternative~1 (as this reduces for a vertex in $V_2$ the number of balanced free edges that are incident with it), hence assume for simplicity that this does not happen either. We shall show that there cannot be many free edges joining two vertices from $V_3$, as otherwise the path partition is not canonical. There are several cases to analyze here.

\vspace{5 mm}
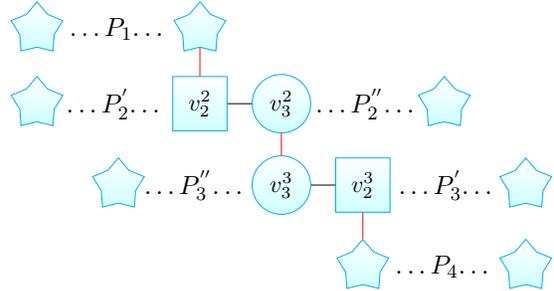
\addfiguretwo One such case replaces four paths $P_1, P_2, P_3, P_4$ in the path partition by only three paths, by making use of a balanced edge between $V_1 \cap P_1$ and $V_2 \cap P_2$, a free edge between $V_3 \cap P_2$ and $V_3 \cap P_3$, and a balanced edge between $V_2 \cap P_3$ and $V_1 \cap P_4$. This process removes one path edge in $P_2$ and one path edge in $P_3$. (See Figure \ref{example:2}).

Given that free edges incident with $V_3$ only rarely have their other endpoint in $V_1 \cup V_2 \cup V_3$, we can infer that $V_4 \cup V_5$ has substantial size.

\vspace{10pt}

The overview provides intuition of why one may hope that the theorem is true. The following section explains how to turn this intuition into a rigorous proof, and moreover, keep the complexity of the proof manageable.

\subsection{Proof of \hyperref[thm:main]{\text{Theorem} \ref{thm:main}}  }

In order to make the proof rigorous, we shall use an accounting method, in which each vertex counts as $1$ point, and points (and also fractions of points) are transferable among vertices. We shall design a collection of point transfer rules, and will prove that using these rules every component ends up with at least~7 points. This implies that the average size of a component is at least~7.

For the purpose of defining the point transfer rules, we refine the classification of vertices introduced earlier.

\begin{itemize}
	
	\item $V_2$ is partitioned into two subclasses depending on whether it has path neighbors in $V_2$:
	
	\begin{itemize}
		\item $V_2^a$ (``a" for ``alone") for those vertices in $V_2$ that do not have a path neighbor in $V_2$.
		\item $V_2^b$ (``b" for ``both") for those vertices in $V_2$ that have a path neighbor in $V_2$.
	\end{itemize}
	
	\item We will identify certain subsets of $V_2$:
	\begin{itemize}
		
		\item A vertex $u\in V_2$ 
		is called {\em moderate} if it has at least two balanced edges, where at least one of them is to a path.
		
		\item A vertex $u\in V_2$ 
		is called {\em heavy} if it has at least three balanced edges to paths. Observe that every heavy vertex is also moderate.
		
	\end{itemize}
	\item A vertex $u \in V_3$ will be called {\em dangerous} 
	if one of its path neighbors is heavy and the other is moderate.
	
\end{itemize}

\nomenclature[v4]{{\em dangerous} $V_3$ }{vertices in $V_3$ that one of its path neighbors is a heavy and the other is moderate.}%
\nomenclature[v0]{$V_2^a$}{vertices in $V_2$ that do not have a path neighbor in $V_2$.}%
\nomenclature[v0]{$V_2^b$}{vertices in $V_2$ that have a path neighbor in $V_2$.}%
\nomenclature[v1]{{\em heavy} $V_2$ }{vertices in $V_2$ that have at least three balanced edges to paths.}%
\nomenclature[v2]{{\em modereate} $V_2$ }{vertices in $V_2$ that have at least two balanced edges where at least one of them is to a path.}%


Every vertex starts with $1$ point, and transfers points using the following point transfer rules. 

\begin{mdframed}[linecolor=black!40,
	outerlinewidth=1pt,
	roundcorner=.5em,
	innertopmargin=1.3ex,
	innerbottommargin=.5\baselineskip,
	innerrightmargin=1em,
	innerleftmargin=0.4em,
	backgroundcolor=blue!10,
	shadow=true,
	shadowsize=6,
	shadowcolor=black!20,
	frametitle={\Large Transfer rules:},
	frametitlebackgroundcolor=cyan!40,
	frametitlerulewidth=14pt
	]
	
	
	\begin{enumerate}[leftmargin=40pt,label=\textbf{Rule \arabic*}
		]
		
		\item From $V_2$ to $V_1\cap \mathcal{C}$. $v\in V_2$ transfers $\frac{1}{i}$ points to $u\in V_1\cap \mathcal{C}$ , if  $u$ is a vertex of a cycle of size $i \le 6$ and $(u,v)$ is a balanced edge. \label{rule:v2_cycle}
		
		\item From $V_2$ to $V_1\cap \mathcal{P}$. $v\in V_2$ transfers $\frac{2}{3}$ points to $u\in V_1$ if $u$ is an end-vertex of a path and $(u,v)$ is a balanced edge. \label{rule:v2_path}
		
		\item From $V_2^a$ to dangerous $V_3$. $v\in V_2^a $ transfers $\frac{1}{6}$ points to $u\in V_3$, if $u$ is dangerous and $(u,v)$ is a free edge.  \label{rule:v2a_dangerous}
		
		\item From $V_2^b\cup V_4$ to dangerous $V_3$. $v\in V_2^b  \cup V_4$ with exactly one path neighbor in $V_2$ transfers $\frac{1}{12}$ points to $u\in V_3$, if $u$ is dangerous  and $(u,v)$ is a free edge.  \label{rule:v4_dangerous}
		
		\item From $V_5$ to $V \setminus V_5$. $v\in V_5$ transfers $\frac{1}{4}$ points to $u\in (V \setminus V_5)$ if $(u,v)$ is a free edge. \label{rule:v5_rest}
		
	\end{enumerate}
\end{mdframed}

We note that due to the transfer rules, a vertex may end up with a negative number of points.

The number of points of a component in a path partition is the sum of points that its vertices have. The point transfer rules immediately imply the following propositions.

\begin{pro}
	\label{pro:easy7}
	Let $C$ be a cycle in a canonical path partition. Then after applying the transfer rules $C$ has at least~7 points.
\end{pro}

\begin{proofing}{Proposition \ref{pro:easy7}}
	{According to the point transfer rules, a vertex of a cycle cannot transfer points to other vertices (because all its vertices are in $V_1$). Let $i$ denote the size of the cycle. If $i \ge 7$ then we are done. If $i < 7$ (and necessarily $i \ge 3$), then each vertex has at least $7-i$ balanced edges connecting it to vertices in $V_2$, because the degree of each vertex is~6. Each balanced edge contributes to the vertex $\frac{1}{i}$ points (by \ref{rule:v2_cycle}),hence, each vertex on the cycle has at least $1 + (7-i)\frac{1}{i}$ points. Hence the cycle has at least $i\cdot(1 + (7-i)\frac{1}{i}) = 7$ points, as desired.}
\end{proofing}

\begin{pro}
	\label{pro:easyouter}
	Let $P$ be a path in a canonical path partition. Then after applying the transfer rules the end-vertices of $P$ have $7+ \frac{5}{3}$ points.
\end{pro}

\begin{proofing}{Proposition \ref{pro:easyouter}}
	{In a canonical path partition, there are no isolated vertices. Hence $P$ has~10 balanced edges incident with its two end-vertices. Applying \ref{rule:v2_path} together with the two starting points, sums up to $2 + 10\cdot \frac{2}{3}=7+ \frac{5}{3}$ points. }
\end{proofing}

\begin{pro}\label{pro:easypoints}
	The only vertices that can end up with negative points are those from $V_2$. In particular:
	\begin{enumerate}
		\item For vertices in $V_2$:
		\begin{enumerate}
			\item A vertex in $V_2$ has at least $-\frac{5}{3}$ point.\label{item:v2} 	
			\item  A vertex in $V_2$ that among its balanced edges, has no edge to a path, ends up with at least $-\frac{1}{3}$ points. \label{item:v2_4cycles}


			\item A vertex in $V_2$ that has among its balanced edges, one edge to a path, ends up with at least $-\frac{2}{3}$ points. \label{item:v2_1path3cycles} 
			\item  A vertex in $V_2$ that has among its balanced edges, two balanced edges to paths, ends up with at least $-1$ points.\label{item:v2_2path2cycles} 
			\item A vertex in $V_2^a$ that has only one balanced edge, ends up with at least $-\frac{1}{6}$ points.\label{item:v2_1path0cycles}

		\end{enumerate}
		
		\item For vertices in $V_3$.
		\begin{enumerate}
			\item A vertex in $V_3$ has at least $1$ point.\label{item:v3}
			\item A vertex in $V_3$ that has all its frees to $V_5$ or $V_2^a$ ends up with at least $\frac{5}{3}$ points.\label{item:v3_good}
		\end{enumerate}
		
		\item Every vertex in $V_4$ has at least $\frac{2}{3}$ points.\label{item:v4}
		\item Every vertex in $V_5$ has a non-negative number of points.\label{item:v5}
	\end{enumerate}
\end{pro}

\begin{proofing}{Proposition \ref{pro:easypoints}}
	{	
		Each vertex starts with $1$ point.
		\begin{enumerate}
			
			\item 	\begin{enumerate}
				
				\item A vertex in $V_2$ has four free edges and may transfer at most $\frac{2}{3}$ points on each free edge, by \ref{rule:v2_path}. Hence it has at least $1 - 4\cdot \frac{2}{3} = -\frac{5}{3}$ points.

				\item  An edge to a cycle can transfer at most $\frac{1}{3}$ points by $\ref{rule:v2_cycle}$. Hence if $x$ has no balanced edges to paths it will end with at least $1 -4\cdot\frac{1}{3}\geq -\frac{1}{3}$ points.

				\item  If among its balanced edges  $x\in V_2$ has exactly one edge to a path, then on this edge, $x$ transfers $\frac{2}{3}$ points (by \ref{rule:v2_path}). The remaining three free edges might be balanced edges to cycles of size three, transferring  $3\cdot \frac{1}{3}$ points (by \ref{rule:v2_cycle}). Hence the number of points on $x$ is at least $1-\frac{2}{3}-3\cdot\frac{1}{3}=-\frac{2}{3}$.
				
				\item If among its balanced edges  $x\in V_2$ has two balanced edges to paths, then these edges can transfer at most $2 \cdot \frac{2}{3}$ points (by \ref{rule:v2_path}). The remaining free edges can be balanced edges to cycles, and may transfer at most $2 \cdot \frac{1}{3}$ points (by \ref{rule:v2_cycle}). Hence, on $x$ there are at least $1- 3\cdot \frac{2}{3}\geq -1$ points.
				
				\item  An edge that is not balanced can transfer at most $\frac{1}{6}$ points from a vertex in $V_2^a$ by $\ref{rule:v2a_dangerous}$. Hence if $x\in V_2^a$ and $x$ has only one balanced edge, it has at least $1-\frac{2}{3} -3\cdot\frac{1}{6}\geq -\frac{1}{6}$ points.

			\end{enumerate}
			\item \begin{enumerate}
				\item Vertices in $V_3$ can only receive points, hence they have at least $1$ point.
				
				\item If $x\in V_3$ has all its frees to $V_5$ or $V_2^a$ then there are two options to consider.
				
				\begin{itemize}
					\item If $x$ has an edge that is incident to $V_5$, then by \ref{rule:v5_rest} it receives $\frac{1}{4}$ points.
					\item If $x$ has an edge that is incident to $V_2^a$, then by \ref{rule:v2a_dangerous} it receives $\frac{1}{6}$ points.
				\end{itemize}
				
				Hence $x$ receives at least $\frac{1}{6}$ points on each free edge and hence ends up with at least $1+4\cdot\frac{1}{6}= \frac{5}{3}$ points.
				
			\end{enumerate}
			\item Vertices in $V_4$ have four free edges, and by \ref{rule:v4_dangerous} they  may transfer no more than $\frac{1}{12}$ points on each free edge. Hence vertices in $V_4$ have at least $\frac{2}{3}$ points.
			\item  Vertices in $V_5$ have four free edges, and may transfer at most $\frac{1}{4}$ points on each free edge (by \ref{rule:v5_rest}).
		\end{enumerate}
	}
\end{proofing}

We now consider the internal vertices of a path, and show that in total they have at least $-\frac{5}{3}$ points. This combined with \hyperref[pro:easyouter]{\textbf{Proposition }\ref{pro:easyouter}} will imply that the path has at least~7 points.

We divide each path $P$ in the canonical partition into disjoint blocks. The blocks need not include all vertices of $P$. We scan the vertices of the path $P$ from left to right. The first block starts at the first vertex from $V_2$ that we encounter (if there is no such vertex, then $P$ has no blocks). Thereafter, we create a sequence $B_1B_2.....B_m$ of blocks as follows. Each block $B_i$ is of the form $B_i=X_iP_i$. Here $X_i$ is a nonempty sequence of consecutive vertices from $V_2$ that is maximal (cannot be extended neither to the left nor to the right). $P_i$ contains those vertices that follow $X_i$, but {\em excluding the vertices that are in $V_5$}. Excluding these vertices is done so as to simplify the presentation. All lower bounds proved on the number of points of $P_i$ without vertices of $V_5$ also hold when vertices of $V_5$ are added back, due to item \ref{item:v5} of \hyperref[pro:easypoints]{\textbf{Proposition} \ref{pro:easypoints}}. $P_i$ ends when either a vertex from $V_2$ is reached (and then this vertex starts a new block), or when the path ends (and then the end vertex from $V_1$ is not included in $X_i$). Hence, in a block $B_i$ that is not last ($i\neq m$), $P_i$ is nonempty and can only be of one of the following two forms: either a single vertex from $V_3$, or a pair of vertices from $V_4$ (because if there were vertices from $V_5$, they were discarded). However, in the last block, $B_m$, it could be that $P_m$ is empty or a single vertex from $V_4$.

We distinguish between four kind of blocks:
\begin{enumerate}
	\item $B_i$ is of Kind $1$ if $|X_i|=1$ and $|P_i|=1$ and $P_i=u\in V_3$.
	\item $B_i$ is of Kind $2$ if $|X_i|=1$ and $|P_i|=2$ (hence vertices of $P_i$ are from $V_4$).
	\item $B_i$ is of Kind $3$ if $|X_i|>1$ and $|P_i|\neq 0$.
	\item $B_i$ is of Kind $4$ if it is the last block of the path. ($|P_i|=0$, or $|P_i|=1$ and $P_i=u\in V_4$.)
\end{enumerate}

Considering this definition of blocks, we state two key lemmas. The proofs of these lemmas show that if the premises of the lemma do not hold, then the path partition is not canonical. The arguments proving this involve a fairly complicated case analysis, and hence the proofs of the lemmas are deferred to after the proof of the main theorem.

\begin{restatable}{lemma}{lem:blockalone}
	\label{lem:blockalone}
	Let $\mathcal{S}$ be a canonical path partition. Let $P$ be a path in $\mathcal{S}$.
	Every block $B_i$ in $P$ has at least $-\frac{5}{3}$ points.
	
\end{restatable}

\begin{restatable}{lemma}{lem:twoblocks}
	\label{lem:twoblocks}
	Let $\mathcal{S}$ be a canonical path partition.  Let $P$ be a path in $\mathcal{S}$. Any sequence of two blocks in $P$, $B_iB_{i+1}$ satisfies at least one of the following three options: \begin{enumerate}
		\item The number points on $B_i$ is non-negative. 
		\item The number points on $B_iB_{i+1}$ is non-negative. 
		\item The number points on $B_iB_{i+1}$ is at least $-1$, and $B_{i+1}$ is of Kind $4$. 
		
	\end{enumerate}

\end{restatable}

Now we can prove \hyperref[thm:main]{\textbf{Theorem} \ref{thm:main}}:

\begin{proofing}{Theorem \ref{thm:main}}
	{Consider a canonical path partition $\mathcal{S}$. We show that applying the set of transfer rules to $\mathcal{S}$ leads to a situation where every component in $\mathcal{S}$ has at least $7$ points. \hyperref[pro:easy7]{\textbf{Proposition} \ref{pro:easy7}} implies that we only need to handle paths. For a path $P$ we know from \hyperref[pro:easyouter]{\textbf{Proposition} \ref{pro:easyouter}} that it receives $7+\frac{5}{3}$ points on its end-vertices. Hence we need to show that the internal vertices have at least $-\frac{5}{3}$ points. All internal vertices up to the beginning of the first block contribute a non-negative number of points (by \hyperref[pro:easypoints]{\textbf{Proposition} \ref{pro:easypoints}}). Thereafter,
		apply options~1 or~2 of \hyperref[lem:twoblocks]{\textbf{Lemma} \ref{lem:twoblocks}} as long as possible, removing blocks with a non-negative contribution to the number of points. Eventually, we are left with either no block (and we are done), one block (and then we can apply \hyperref[lem:blockalone]{\textbf{Lemma} \ref{lem:blockalone}}) or two blocks (with option~3 of \hyperref[lem:twoblocks]{\textbf{Lemma} \ref{lem:twoblocks}} applicable). In any case, we get that for the sequence of blocks we lose at most $\frac{5}{3}$ points, which leaves the path with at least $7$ points at the end.
	}
\end{proofing}


\section{Proofs of \hyperref[lem:blockalone]{\text{Lemma} \ref{lem:blockalone}} and  \hyperref[lem:twoblocks]{\text{Lemma} \ref{lem:twoblocks}}} \label{sec:key_lemmas}
In this section we prove the two key lemmas, \hyperref[lem:blockalone]{\textbf{Lemma} \ref{lem:blockalone}} and  \hyperref[lem:twoblocks]{\textbf{Lemma} \ref{lem:twoblocks}}.
In order to prove these two lemmas, we state a set of lemmas that we will prove in a later section.

\begin{lemma}\label{lem:dangerous}
	Let  $\mathcal{S}$ be a canonical path partition. Let $P$ be a path in $\mathcal{S}$.
	Let $x_1v_3$ ($v_3x_1$, respectively) be two consecutive vertices in $P$, where $x_1$ is heavy and $v_3$ is dangerous. If $(v_3,u)\in E_F$ is a free edge and $u$ has a path neighbor in $V_2$ ($u\in V_4\cup V_2^b\cup V_3 $), then:
	\begin{enumerate}
		\item $u$ is on the same path as $v_3$.\label{itm:1}
		\item $u$ is on the right (left, respectively) of $v_3$ and has only one path neighbor in $V_2$ (hence $u\in V_4\cup V_2^b $), which we denote by $x_2$. $x_2$ is right (left, respectively) of $u$ in the path.\label{itm:2}
		
		\item $x_2$ has exactly one balanced edge among its four free edges. This edge is incident to the end-vertex of the path that is closer to $x_2$ than to $u$. We will denote this end-vertex by $o_2$ ($o_1$, respectively).\label{itm:3}
		
		\item $y_1\in V_2$ which is the other path neighbor of $v_3$ is moderate with exactly one balanced edge to $o_2$ ($o_1$, respectively). The rest of its free edges may be balanced edges to cycles.\label{itm:4}
		
	\end{enumerate}
	
\end{lemma}
\setlength{\intextsep}{12pt}%
\begin{figure}[H]
	
	\begin{subfigure}[t]{0.5\columnwidth}  \resizebox{\linewidth}{!}
		{\dangerouslemma}
		\caption{$x_1v_3$ } \label{fig1:subfig1}
	\end{subfigure}
	\hspace{10mm}
	\begin{subfigure}[t]{0.5\columnwidth}  \resizebox{\linewidth}{!}
		{\dangerouslemmatwo}
		\caption{$v_3x_1$} \label{fig1:subfig2}
	\end{subfigure}
	\caption{The two cases of \hyperref[lem:dangerous]{\text{Lemma} \ref{lem:dangerous}} }\label{fig1}
\end{figure}

\begin{lemma}\label{lem:kadjacent}
	Let $\mathcal{S}$ be a canonical path partition and let $P$ be a path in $\mathcal{S}$. Let $X_i$ be a part of some block $B_i=X_iP_i$ in path $P$ such that $|X_i|=k>1$.
	Then $X_i$ ends with at least $\frac{1}{3}\cdot k-\frac{4}{3}$ points.

\end{lemma}

\begin{lemma}\label{lem:special_case}
	Let $\mathcal{S}$ be a canonical path partition and let $P$ be a path in $\mathcal{S}$. Let $B_iB_{i+1}$ be two consecutive blocks in path $P$ such that $B_i$ is of Kind $1$ ($X_i=x_i$) and $|X_{i+1}|=2$.
	If $x_i$ is heavy then $X_{i+1}$ ends with at least $-\frac{1}{3}$ points.
	
\end{lemma}

\subsection{Proof of \hyperref[lem:blockalone]{\text{Lemma} \ref{lem:blockalone}} }

\begin{proofing}{Lemma \ref{lem:blockalone}}{
		Let $B_i=X_iP_i$ be a block of a path. By \hyperref[pro:easypoints]{\text{Proposition} \ref{pro:easypoints}}, every vertex in $P_i$ contributes a non-negative number of points, and hence $P_i$ contributes a non-negative number of points. As to $X_i$, there are two cases to consider.
		\begin{enumerate}
			\item If $|X_i| = 1$ then the vertex $x_i$ that makes up $X_i$ has at least $-\frac{5}{3}$ points by item \ref{item:v2} of \hyperref[pro:easypoints]{\text{Proposition} \ref{pro:easypoints}}.
			
			\item	If $|X_i|=k >1$ then by \hyperref[lem:kadjacent]{\text{Lemma} \ref{lem:kadjacent}} the number of points is at least $\frac{k}{3}-\frac{4}{3}\geq -\frac{2}{3}$.
		\end{enumerate}

	}
\end{proofing}

\subsection{Proof of \hyperref[lem:twoblocks]{\text{Lemma} \ref{lem:twoblocks}}}

We state two observations that prove some special cases of  \hyperref[lem:twoblocks]{\text{Lemma} \ref{lem:twoblocks}}.

\setcounter{lemma}{2}
\begin{obs}\label{obs:kindtwoblock}
	If block $B_i$ is of Kind $2$ then the number of points on $B_i$ is non-negative.
	
\end{obs}

\begin{proofing}{Observation \ref{obs:kindtwoblock}}
	{Let $B_i=X_iP_i$ be a block of Kind $2$.
		Let $X_i=x_i$ and $P_i=v_4^av_4^b$, where $v_4^a,v_4^b\in V_4$.
		By \ref{rule:v4_dangerous} of the points transfer rules, vertices from $V_4$ can transfer points among free edges to {\em dangerous} vertices. Hence we split the proof into two cases (see Figure \ref{fig2}):
		
		\begin{enumerate}[leftmargin=40pt,label=(\roman*)]
			\item If $v_4^a$ has no free edge to a dangerous vertex then $v_4^a$ ends up with $1$ point. As to the rest of the block, from items \ref{item:v2} and \ref{item:v4} of \hyperref[pro:easypoints]{\text{Proposition} \ref{pro:easypoints}}, $x_i$ and $v_4^b$ have at least $-\frac{5}{3}$ and $\frac{2}{3} $ points respectively. Hence in total we get that $X_iP_i$ ends up with at least  $ -\frac{5}{3}+\frac{2}{3}+1\geq 0 $ points. 
			
			\item If $v_4^a$ has a free edge that is incident to a dangerous vertex $v_3$ then by \hyperref[lem:dangerous]{\text{Lemma} \ref{lem:dangerous}} (case as in Figure \ref{fig1:subfig2}), $v_3$ must be on the same path as $v_4^a$ and must be on the \textbf{right} of it and $x_i$ will have exactly one balanced edge among its four free edges. Hence by \ref{item:v2_1path0cycles} of \hyperref[pro:easypoints]{\text{Proposition} \ref{pro:easypoints}}, $x_i$ ends up with at least $-\frac{1}{6}$ points. By item \ref{item:v4} of \hyperref[pro:easypoints]{\text{Proposition} \ref{pro:easypoints}} each one of the vertices of $P_i$ has at least $\frac{2}{3}$ points and hence $X_iP_i$ will have at least $-\frac{1}{6}+\frac{4}{3}>1$ points.  
			
			
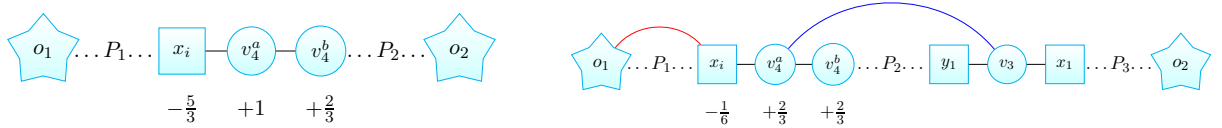
\begin{figure}[H]
				\centering
				\begin{subfigure}[t]{0.45\columnwidth} \centering \resizebox{0.9\linewidth}{!}
					{\blockkindtwo{$-\frac{5}{3}$}{$+1$}{$+\frac{2}{3}$}}
					\caption{$v_4^a$ has no free edge to a dangerous vertex } \label{fig2:subfig1}
				\end{subfigure}
				\hspace{2mm}
				\begin{subfigure}[t]{0.5\columnwidth}
					\centering \resizebox{1.05\linewidth}{!}
					{\blockkindtwo{$-\frac{1}{6}$}{$+\frac{2}{3}$}{$+\frac{2}{3}$}}
					\caption{$v_4^a$ has a free edge to a dangerous vertex $v_3$} \label{fig2:subfig2}
				\end{subfigure}
				\caption{\hyperref[obs:kindtwoblock]{\text{Ovservation} \ref{obs:kindtwoblock}} - Block of Kind $2$}\label{fig2}
			\end{figure}
			
		\end{enumerate}

	}
\end{proofing}

\begin{obs}\label{obs:kindthreeblock}
	If block $B_i$ is of Kind $3$ then number of points on $B_i$ is non-negative. (In fact, at least $\frac{1}{3}$.)
	
\end{obs}

\begin{proofing}{Observation \ref{obs:kindthreeblock}}
	{Let $B_i=X_iP_i$ be of Kind $3$ where $|X_i|=k > 1$. By \hyperref[lem:kadjacent]{\text{Lemma} \ref{lem:kadjacent}} the number of points on $X_i$ is at least $ \frac{k}{3}-\frac{4}{3}\geq -\frac{2}{3}$. As to $P_i$, there are two cases (see Figure \ref{fig3}).
		\begin{enumerate}[label=(\roman*)]
			\item If $|P_i| = 1$ then the vertex of $P_i$ is in $V_3$, and by item \ref{item:v3} of \hyperref[pro:easypoints]{\text{Proposition} \ref{pro:easypoints}} it has at least $1$ point. 	
			\item	If $|P_i|=2$ then the two vertices of $P_i$ are from $V_4$. By item \ref{item:v4} of \hyperref[pro:easypoints]{\text{Proposition} \ref{pro:easypoints}} each one of them has at least $\frac{2}{3}$ points.	

			\begin{figure}[H]{
					\centering
					\begin{subfigure}[t]{0.3\columnwidth}
						\centering
						\resizebox{0.6\linewidth}{!}
						{\blockkindthree{$2$}}
						\caption{$|P_i| = 1$ } \label{fig3:subfig1}
					\end{subfigure}
					\hspace{10mm}
					\begin{subfigure}[t]{0.3\columnwidth} 	
						\centering
						\resizebox{0.8\linewidth}{!}
						{\blockkindthree{$1$}}
						\caption{$|P_i| = 2$} \label{fig3:subfig2}
					\end{subfigure}
					\caption{\hyperref[obs:kindthreeblock]{\text{Observation} \ref{obs:kindthreeblock}} - Block of Kind $3$}\label{fig3}
				}	
			\end{figure}
		\end{enumerate}
		In either case, the total number of points on $B_i$ is at least $\frac{1}{3}$.
	}
\end{proofing}

Using the above two observations we now prove \hyperref[lem:twoblocks]{\text{Lemma} \ref{lem:twoblocks}}.

\begin{proofing}{Lemma \ref{lem:twoblocks}}
	{
		By \hyperref[obs:kindthreeblock]{\text{Observation} \ref{obs:kindthreeblock}} and  \hyperref[obs:kindtwoblock]{\text{Observation} \ref{obs:kindtwoblock}} we get that the only case to deal with is when $B_i$ is of Kind $1$. In this case it suffices to show that if the number of points on $B_i$ is negative, then the number of points on $B_iX_{i+1}$ is at least $-1$. This suffices by the following two cases. If $B_{i+1}$ is of kind $4$, the number of points on $B_iB_{i+1}$ is at least as the number of points on $B_iX_{i+1}$, proving option~$3$ in the lemma. If $B_{i+1}$ is not of kind $4$, then $P_{i+1}$ is not empty and has at least one point (as in the proof of \hyperref[obs:kindthreeblock]{\text{Observation} \ref{obs:kindthreeblock}}), implying that the number of points on $B_iB_{i+1}$ is non-negative.
		

		
		
		In order to prove the sufficient condition above, we now split the proof into cases according to $X_i$.
		Let $X_i=x_i$ and $P_i=v_i$, where $v_i\in V_3$.
		\begin{enumerate}
			\item If $x_i$ is not heavy then among its four free edges it has at most two balanced edges to paths. Hence by item \ref{item:v2_2path2cycles} of \hyperref[pro:easypoints]{\text{Proposition} \ref{pro:easypoints}}, it has at least $-1$ points. By item \ref{item:v3} of \hyperref[pro:easypoints]{\text{Proposition} \ref{pro:easypoints}}, $v_i$ has at least $1$ point. Hence the number of points on $B_i$ is non-negative. (See Figure \ref{fig4}.)
			
			
			\begin{figure}[H]{
					\centering
					\resizebox{0.2\linewidth}{!}
					{\blockkindone{$-1$}{$+1$}}

					\caption{Block of Kind $1$ when $x_i$ is not heavy}\label{fig4}
				}	
			\end{figure}

			\item If $x_i$ is heavy then from item \ref{item:v2} of \hyperref[pro:easypoints]{\text{Proposition} \ref{pro:easypoints}} $x_i$ has at least $-\frac{5}{3}$ points. We now consider two cases depending on the size of $X_{i+1}$.
			\begin{enumerate}
				\item If $|X_{i+1}|=1$, let $x_{i+1}$ be the vertex that makes up $X_{i+1}$. We split this case into three cases. (See Figure \ref{fig5}.)
				\begin{enumerate}
					\item If $x_{i+1}$ is moderate and has at least two balanced edges to paths, then $v_i$ (which is dangerous) has all its free edges to $V_5$ or $V_2^a$. This is because if $v_i$ has an edge that is incident to a vertex that has a path neighbor in $V_2$, then by \hyperref[lem:dangerous]{\text{Lemma} \ref{lem:dangerous}} on $x_{i}v_i$, $x_{i+1}$ cannot have more than one balanced edge to a path, contradiction.
					Hence, by item \ref{item:v3_good} of \hyperref[pro:easypoints]{\text{Proposition} \ref{pro:easypoints}}, $v_i$ ends with at least $\frac{5}{3}$ points. Thus, $B_i$ is non-negative. 	
					\item If $x_{i+1}$ is moderate and has exactly one balanced edge to paths, then by item \ref{item:v2_1path3cycles} of \hyperref[pro:easypoints]{\text{Proposition} \ref{pro:easypoints}}  we get that the number of points on $x_{i+1}$ is at least $-\frac{2}{3}$. As to the free edges incident with $v_i$, there are three options.

					\begin{itemize}
						\item If $v_i$ has a free edge that is incident to a vertex that has a path neighbor in $V_2$ then by \hyperref[lem:dangerous]{\text{Lemma} \ref{lem:dangerous}} on $x_{i}v_i$ it must be incident to $V_4\cup V_2^b$. Hence by \ref{rule:v4_dangerous} it receives $\frac{1}{12}$ points.
						\item If $v_i$ has an edge that is incident to  $V_2^a$, then  by \ref{rule:v2a_dangerous} it receives $\frac{1}{6}$ points.
						
						\item If $v_i$ has an edge that is incident to $V_5$, then  by \ref{rule:v5_rest} it receives $\frac{1}{4}$ points.
						
					\end{itemize}
					
					Therefore, $v_i$ receives at least $\frac{1}{12}$ points on each free edge and hence has at least $1+4\cdot\frac{1}{12}= \frac{4}{3}$ points. Thus, the number of points on $B_iX_{i+1}$ is at least $-\frac{5}{3}+\frac{4}{3}-\frac{2}{3}\ge -1$.
					
					If $x_i$ is not heavy then among its four free edges it has at most two balanced edges to paths.
					
					%

					\item If $x_{i+1}$ is not moderate then there are two options. Either its all balanced edges are to cycles or it has exactly one balanced edge, and this edge is incident to a path.
					\begin{itemize}
						\item If $x_{i+1}$ has no balanced edges to paths, then it ends up with at least $-\frac{1}{3}$ points by item \ref{item:v2_4cycles} of \hyperref[pro:easypoints]{\text{Proposition} \ref{pro:easypoints}}.
						\item	If $x_{i+1}$ has exactly one balanced edge then by item \ref{item:v2_1path0cycles} of \hyperref[pro:easypoints]{\text{Proposition} \ref{pro:easypoints}} it ends up with at least $-\frac{1}{6}$ points.
					\end{itemize}

					Therefore, the number of points on $B_iX_{i+1}$ is at least $-\frac{5}{3}+1-\frac{1}{3}\geq -1$.
					
					
					
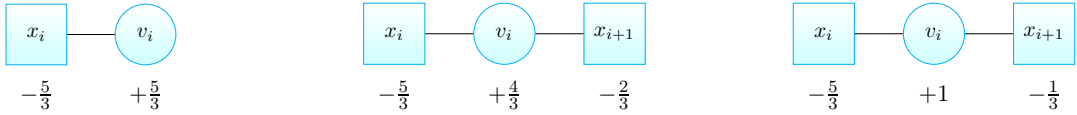
\begin{figure}[H]
						\centering
						\begin{subfigure}[t]{0.25\columnwidth} 	
							\centering
							\resizebox{0.6\linewidth}{!}{\blockkindoneheavy{a}}
							\caption{\small If $x_{i+1}$ is moderate and has at least two balanced edges to paths }\label{fig5:subfig1}
						\end{subfigure}
						\hspace{7mm}
						\begin{subfigure}[t]{0.3\columnwidth} 	
							\centering
							\resizebox{0.8\linewidth}{!}
							{\blockkindoneheavy{b}}
							\caption{\small If $x_{i+1}$ is moderate and has only one balanced edge to a path} \label{fig5:subfig2}
						\end{subfigure}
						\hspace{5mm}
						\begin{subfigure}[t]{0.3\columnwidth} 	
							
							\centering
							\resizebox{0.8\linewidth}{!}
							{\blockkindoneheavy{c}}
							\caption{If $x_{i+1}$ is not moderate } \label{fig5:subfig3}
						\end{subfigure}
						\caption{\small $x_i$ is heavy and $X_{i+1}$ is of size $1$}\label{fig5}
					\end{figure}
					
				\end{enumerate}
				\item If $|X_{i+1}|=k > 1$, then we show that for any $k\geq 2$,  $X_{i+1}$ has at least $-\frac{1}{3}$ points. Taking it together with the fact that on $B_i$ there are at least $-\frac{5}{3}+1\geq -\frac{2}{3}$ points, we get that on $B_iX_{i+1}$ there are at least $-\frac{2}{3}-\frac{1}{3}\geq -1$ points. (see Figure \ref{figheavy}.)
				
				
				\begin{enumerate}
					\item If $k=2$ then because $x_i$ is heavy, \hyperref[lem:special_case]{\text{Lemma} \ref{lem:special_case}} implies that $X_{i+1}$ has at least $-\frac{1}{3}$ points.
					\item 	If $k>2$ then by \hyperref[lem:kadjacent]{\text{Lemma} \ref{lem:kadjacent}}, $X_{i+1}$ has at least $ \frac{1}{3}\cdot k-\frac{4}{3}\geq -\frac{1}{3}$ points for $k\geq 3$.
					
				\end{enumerate}

				\begin{figure}[H]
					
					\centering
					\resizebox{0.2\linewidth}{!}
					{\blockkindoneend}

					\caption{$x_i$ is heavy and $X_{i+1}$ is not of size $1$}\label{figheavy}
				\end{figure}
				
			\end{enumerate}
		\end{enumerate}

	}	
\end{proofing}

\section{Properties of a canonical path partition} \label{sec:canonical}

We first prove that indeed a canonical path partition as defined on page~\pageref{can} exists. The following Lemma is based on a similar lemma that appears in~\cite{MagnantM09}.

\setcounter{lemma}{5}
\begin{lemma}\label{lem:canonical}
	For $d>0$, every $d$-regular graph $G$ has a canonical path partition.
\end{lemma}

\begin{proofing} {Lemma \ref{lem:canonical}}{
		We show that given that Properties 1 and 2 of canonical path partitions hold, Property 3 can be enforced to hold as well.
		Let $\mathcal{S}$ be a path partition satisfying Properties $1$ and $2$ with as few isolated vertices as possible. We refer to such an $\mathcal{S}$ as a pseudo-canonical path partition. Our goal is to show that a pseudo-canonical path partition has no isolated vertices, making it a canonical path partition.
		
		Assume towards contradiction that there is a vertex $v$ that forms a component in $\mathcal{S}$. We define a collection $\mathcal{F}$ of components of $\mathcal{S}$ by the following process. We add to $\mathcal{F}$ all those components that are incident to $v$ by free edges. These components must be paths, by Property $1$. Now recursively, in an arbitrary order, for each such path $P$, add to $\mathcal{F}$ all paths that are incident to an end-vertex of $P$ (if they were not added previously).
		
		This process must end because the graph is finite.

		\begin{obs}\label{sizethree}
			Every path that is added to $\mathcal{F}$ is of size three.
		\end{obs}
		
		\begin{proofing}{Observation \ref{sizethree}}
			{
				Let $v$ be an isolated vertex in a pseudo-canonical path partition. Note that $v$ cannot be incident to a vertex on a cycle or to an end-vertex of a path because of Property $1$. Hence if $v$ is incident to a vertex $y$, then $y$ must be an internal vertex of a path $P$. Moreover, $P$ cannot have more than one internal vertex, because then we could make two paths out of $v$ and $P$, with at least two vertices each, contradicting the minimality of the number of isolated vertices. Hence $P$ must be a path of size three. Let $P=\{x,y,z\}$. We claim that every free edge of an end point of $P$ (say, $x$) is incident only to vertices that are internal vertices of paths of size three in $\mathcal{S}$. For $x$ one can easily create a new path partition, in which $x$ is isolated, by adding the edge $(v,y)$. This simple transformation will create a path partition that has the same number of components, cycles and isolated vertices as in $\mathcal{S}$, and hence it too is pseudo-canonical.  As a result, repeating this process recursively we add to $\mathcal{F}$ only paths of size three.
			}
		\end{proofing}

		We denote by $O$ the end-vertices of the paths in $\mathcal{F}$ and by $I$ the inner vertices.
		Every path in $\mathcal{F}$ has two vertices in $O$ and one vertex in $I$. Hence, we get that $|I|\leq \frac{|O|}{2} \implies |I|< |O|$.
		By Properties~1 and~2, there are no edges between vertices of $O$. Moreover, when the process of constructing $\mathcal{F}$ ends, $O$ has no neighbor outside $\mathcal{F}$. Hence, all $d$ neighbors of vertices in $O$ are in $I$. On the other hand, any vertex in $I$ can be incident to at most $d$ vertices of $O$. We get that  $d|O|\leq d|I|\implies |O|\leq |I|$, and this is a contradiction.
		
		
	}
\end{proofing}


\subsection{Properties of $V_3$ and the proof of \hyperref[lem:dangerous]{\text{Lemma} \ref{lem:dangerous}}}  \label{subsec:pro3}

In this section we prove \hyperref[lem:dangerous]{\text{Lemma} \ref{lem:dangerous}}. We show that if the premises of the lemma do not hold, then the path partition is not canonical. To prove \hyperref[lem:dangerous]{\text{Lemma} \ref{lem:dangerous}} we will use the following scheme. We start with a canonical path partition $\mathcal{S}$ and temporarily create from $\mathcal{S}$ a new path partition $\mathcal{S}_0$ with one more component. Then we show that we can create from $\mathcal{S}_0$ a new path partition that either has fewer components than $\mathcal{S}$, or has the same number of components as $\mathcal{S}$, but more cycles. Both of these two cases imply that $\mathcal{S}$ is not canonical.

To simplify the presentation, some components in $S_0$ will be referred to as {\em pseudo-paths}. These are components that are treated as paths but might actually be cycles, and their nature becomes apparent only in later stages. Hence $S_0$ will have components that are cycles, components that are paths, and components that are pseudo-paths.
We say that a path partition $\mathcal{S}_0$ is {\em derived} from $\mathcal{S}$ if: \label{def:derived_path}
\begin{itemize}
	\item $\mathcal{S}_0$ is a path partition of size $|\mathcal{S}_0|=|\mathcal{S}|+1$
	
	\item  $\mathcal{S}$ and $\mathcal{S}_0$ have the same cycle components.
	
	\item All the end-vertices of paths in $\mathcal{S}$ are end-vertices of either paths or pseudo-paths in $\mathcal{S}_0$.
	
	\item The two new end-vertices (of paths or pseudo-paths) in $\mathcal{S}_0$ (that are not end-vertices in $\mathcal{S}$) are in $V_2$ of $\mathcal{S}$.
	
\end{itemize}

\begin{lemma}\label{lem:derived}
	Let  $\mathcal{S}_0$ be a path partition that is derived from $\mathcal{S}$. If one of the new end-vertices of $\mathcal{S}_0$ is heavy in $\mathcal{S}$, then $\mathcal{S}$ is not canonical.
\end{lemma}

\begin{proofing}{Lemma \ref{lem:derived}}
	{	Let  $\mathcal{S}_0$ be a path partition that is derived from $\mathcal{S}$. Let $x_1$ and $x_2$ be the new end-vertices of $\mathcal{S}_0$, where $x_1,x_2$ are in $V_2$ in $\mathcal{S}$, and moreover, $x_1$ is heavy in $\mathcal{S}$.
		Let $(x_2,o_{x_2})\in E$ be a balanced edge incident to $x_2$ (there must be at least one such edge). Note that $o_{x_2}$ is in $V_1$ of $\mathcal{S}$, and hence it is either a cycle vertex of $\mathcal{S}_0$, or an end vertex of either a path or a pseudo-path of $\mathcal{S}_0$, but $o_{x_2} \not= x_1$.
		One should consider the following cases (see Figure \ref{fig11}):
		\begin{enumerate}[leftmargin=40pt, label=     \textbf{Case \arabic*}]
			\item If $x_1$ and $x_2$ are the end-vertices of the same path or pseudo-path, $Q_1$, in $\mathcal{S}_0$, then $o_{x_2}$ is on a component $P_{x_2}\neq Q_{1}$ (which can be a cycle). As $x_1$ is heavy (has at least three balanced edges to paths), there must be a balanced edge $(x_1, o_{x_1})\in E$ where $o_{x_1}$ is on
			$P_{x_1}\neq P_{x_2}$. We then concatenate $Q_1$, $P_{x_2}$ and $P_{x_1}$ into one path, making a new path partition of size  $|\mathcal{S}|-1$, contradicting Property~1. (See Figure \ref{fig11:subfig1}.)
			
			\item If $x_1$ and $x_2$ are on two different paths or pseudo-paths, $Q_1=o_1P_1x_1$ and $Q_2=o_2P_2x_2$ in $\mathcal{S}_0$, then there are three cases to consider.\label{case2} (See Figure  \ref{fig11:subfig2}.)
			
			\begin{enumerate}
				
				\item If $o_{x_2}$ is an end-vertex of $Q_1$ or $Q_2$, then because $x_1$ is heavy (has at least three balanced edges to paths), there must be a balanced edge connecting $x_1$ to an end-vertex of some path  $P_{x_1}\neq Q_1,Q_2$. (See Figure \ref{fig11:subfig2a}).
				\begin{enumerate}
					\item If $o_{x_2}=o_2$ is on $Q_2$, 
					then the pseudo-path $Q_2$ is a cycle. Concatenate $Q_1$ and $P_{x_1}$ into one path leads to a new path partition of size $|\mathcal{S}|$ that contains one more cycle than $\mathcal{S}$ (because $P_{x_1}$ is not a cycle), contradicting Property~2. (See blue in \ref{fig11:subfig2a}.)
					\item If $o_{x_2}=o_1$ on $Q_1$, 
					we concatenate $Q_2$, $Q_1$ and $P_{x_1}$ into one path, making a new path partition of size  $|\mathcal{S}|-1$, contradicting Property~1. (See green in Figure \ref{fig11:subfig2a}.)
					
				\end{enumerate}
				\item  If $o_{x_2}$ is on a cycle component $C$, we consider a balanced edge $(x_1, o_{x_1})\in E$ where $o_{x_1}$ is an end vertex of a path $P_{x_1}$ (such an edge must exist because $x_1$ is heavy and consequently has at least three balanced edges to paths). We concatenate $Q_1$ and $P_{x_1}$ into one path and $C$ and $Q_2$ into anther path, making a new path partition of size  $|\mathcal{S}|-1$, contradicting Property~1. (See Figure \ref{fig11:subfig2b}.)
				\item If $o_{x_2}$ is on a path $P_{x_2}\neq Q_{1},Q_2$, 
				then consider a balanced edge $(x_1, o_{x_1})\in E$ incident to $x_1$, with $o_{x_1}\neq o_{x_2}$ being an end-vertex of a path (such an edge must exist because $x_1$ is heavy). There are three cases to consider. (See Figure \ref{fig11:subfig2c}.)
				\begin{enumerate}
					\item If $o_{x_1}$ is an end-vertex of $Q_1$ or $Q_2$ then we treat in a way similar to Case 2(a). (See blue in Figure \ref{fig11:subfig2c}.)
					\item If $o_{x_1}$ is on the other side of $P_{x_2}$ then we create one path out of $Q_1$, $P_{x_2}$ and $Q_2$, making a new path partition of size  $|\mathcal{S}|-1$, contradicting Property~1.(See green in Figure \ref{fig11:subfig2c}.)
					\item If $o_{x_1}$ is on a path $P_{x_1}\neq P_{x_2}$, then create one path out of $Q_1$ and $P_{x_1}$, and one path out of $Q_2$ and $P_{x_2}$, making a new path partition of size  $|\mathcal{S}|-1$, contradicting Property~1. (See red in Figure  \ref{fig11:subfig2c}.)
					
				\end{enumerate}
				
			\end{enumerate}

		\end{enumerate}

		\begin{figure}[H]
			\centering
			\begin{subfigure}{0.4\textwidth}
				\centering
				\resizebox{\linewidth}{!}{\caseexternaltwo}
				\caption{Case 1 - $x_1$ and $x_2$ are on the same path or pseudo-path $Q_1$.}
				
				\label{fig11:subfig1}
			\end{subfigure}%
			\vspace{10mm}
			\begin{subfigure}{\textwidth}
				\centering
				\begin{subfigure}{0.27\textwidth}
					\renewcommand\thesubfigure{\roman{subfigure} a}
					\centering
					\resizebox{\linewidth}{!}{\caseexternala}
					\caption{\centering $o_{x_2}$ is an end-vertex of $Q_1$ or $Q_2$}
					\label{fig11:subfig2a}
				\end{subfigure}
				\hspace{10mm}
				\begin{subfigure}{0.27\textwidth}	
					\addtocounter{subfigure}{-1}
					\renewcommand\thesubfigure{\roman{subfigure} b}
					\centering
					\resizebox{\linewidth}{!}{\caseexternalb}
					\caption{\centering $o_{x_2}$ is on a cycle  $C$ }
					\label{fig11:subfig2b}
				\end{subfigure}
				\hspace{10mm}
				\begin{subfigure}{0.27\textwidth}	
					\addtocounter{subfigure}{-1}
					\renewcommand\thesubfigure{\roman{subfigure} c}
					\centering
					\resizebox{\linewidth}{!}{\caseexternalc}
					\caption{\centering $o_{x_2}$ is on a component $P_{x_2}\neq Q_{1},Q_2$}
					\label{fig11:subfig2c}
				\end{subfigure}
				
				\addtocounter{subfigure}{-1}
				\caption{Case 2 - $x_1$ and $x_2$ are on different paths or pseudo-paths. } \label{fig11:subfig2}
			\end{subfigure}
			\caption{\hyperref[lem:derived]{\text{Lemma} \ref{lem:derived}} where $x_1$ is heavy }
			\label{fig11}
		\end{figure}

	}

\end{proofing}

The following observations consider situations where a dangerous vertex, $v_3$, has a free edge to a vertex $u$. All of them are proven easily using \hyperref[lem:derived]{\text{Lemma} \ref{lem:derived}}.
\setcounter{lemma}{3}
\begin{obs}\label{different}
	Let $\mathcal{S}$ be a canonical path partition. Let $(u,v_3)\in E_F$ be a free edge, such that $v_3$ is dangerous. If $u$ is not on the same path as $v_3$, then $u$ has no path neighbor in $V_2$.
\end{obs}

\begin{proofing}{Observation \ref{different}}
	{Let $\mathcal{S}$ be a canonical path partition. Let $(u,v_3)\in E_F$ be a free edge, such that $v_3$ is dangerous and $u$ is not on the same path as $v_3$. Assume towards contradiction that $u$ has a path neighbor $x_2\in V_2$. Let $P_1=P_1'x_1v_3y_1P_1^{''}$ be the path that $v_3$ lies on ($x_1$ is the heavy neighbor of $v_3$). Let $P_2=P_2'x_2uP_2^{''}$ be the path that $u$ lies on ($x_2\in V_2$). (See Figure \ref{fig7:subfig1}.) We can then create from $P_1$ and $P_2$ the following three paths: $Q_1= P_1'x_1$, $Q_2=P_2'x_2$ and $Q_3=P_2^{''}uv_3y_1P_1^{''}$. From the union of  $\{Q_1,Q_2,Q_3\}$ and $\mathcal{S}\backslash \{P_1, P_2\}$ we create a path partition $\mathcal{S}_0$. (See Figure \ref{fig7:subfig2}.) Note that $\mathcal{S}_0$ is a derived path partition from $\mathcal{S}$ (as defined on page~\pageref{def:derived_path}). In this derived path partition, $x_1$, which is one of the ``new" end vertices, is heavy. Hence by \hyperref[lem:derived]{\text{Lemma} \ref{lem:derived}}, $\mathcal{S}$ is not canonical, which is a contradiction to $\mathcal{S}$ being canonical.
		\begin{figure}[H]
			\centering
			\begin{subfigure}[t]{0.4\columnwidth}  \resizebox{\linewidth}{!}
				{		\difpath{$x_1$}{$v_3$}{$u$}{$x_2$}}
				\caption{Before } \label{fig7:subfig1}
			\end{subfigure}
			\hspace{15mm}
			\begin{subfigure}[t]{0.4\columnwidth}  \resizebox{\linewidth}{!}
				{\difshiled{$x_1$}{$x_2$}{$u$}{$v_3$}}
				\caption{After} \label{fig7:subfig2}
			\end{subfigure}
			\caption{\hyperref[different]{\text{Observation} \ref{different}}} \label{fig7}
		\end{figure}
	}
	
\end{proofing}

\begin{obs}\label{samesides}
	Let $P= P_1x_1v_3y_1P_2x_2uP_3$ ($x_1,x_2$ are on the left side) be a path in a canonical path partition $\mathcal{S}$, where $(v_3,u)\in E_F$ and  $x_1$ is heavy. Then $x_2\not \in V_2$. Likewise, if the premises hold for $P= P_1y_1v_3x_1P_2ux_2P_3$ ($x_1,x_2$ are on the right side), then $x_2\not \in V_2$.
\end{obs}

\begin{proofing}{Observation \ref{samesides}}
	{Let $P= P_1x_1v_3y_1P_2x_2uP_3$, where $(v_3,u)\in E_F$ and $x_1$ is heavy. (see Figure \ref{fig8}). Assume towards contradiction that $x_2\in V_2$. We can then create the following two paths from $P$: $Q_1=P_3uv_3P_2x_2$ and $Q_2=P_1x_1$. From the union of $\{Q_1,Q_2\}$ and $\mathcal{S}\backslash P$ we create a path partition $\mathcal{S}_0$. Note that $\mathcal{S}_0$ is a derived path partition from $\mathcal{S}$. 
		In this derived path partition, $x_1$, which is one of the ``new" end vertices, is heavy. Hence by \hyperref[lem:derived]{\text{Lemma} \ref{lem:derived}}, $\mathcal{S}$ is not canonical, contradicting the premises of the lemma.
		
		The case $P= P_1y_1v_3x_1P_2ux_2P_3$ is handled in an analogous way.
		
		\begin{figure}[!h]
			\centering
			\scalebox{.8}{\samesides{$x_1$}{$v_3$}{$x_2$}{$u$}}
			\caption{\hyperref[samesides]{\text{Observation}  \ref{samesides}}}\label{fig8}
		\end{figure}
	}
\end{proofing}

\begin{obs}\label{between}
	Let  $P= P_1ux_2P_2x_1v_3y_1P_3$ ($x_1,x_2$ are between) be a path in a canonical path partition $\mathcal{S}$, where $(v_3,u)\in E_F$ and $x_1$ is heavy. Then $x_2\not\in V_2$.
\end{obs}

\begin{proofing}{Observation \ref{between}}
	{	Let  $P= P_1ux_2P_2x_1v_3y_1P_3$ be a path in a canonical path partition $\mathcal{S}$, where $(v_3,u)\in E_F$ and $x_1$ is heavy (see Figure \ref{fig9}). Assume towards contradiction that $x_2\in V_2$. We can then create the following two paths from $P$: $Q_1=x_1P_2x_2$, $Q_2=P_1uv_3y_1P_3$. From the union of  $\{Q_1,Q_2\}$ and $\mathcal{S}\backslash P$ we create a path partition $\mathcal{S}_0$. Note that $\mathcal{S}_0$ is a derived path partition from $\mathcal{S}$. 
		In this derived path partition, $x_1$, which is one of the ``new" end vertices, is heavy. Hence by \hyperref[lem:derived]{\text{Lemma} \ref{lem:derived}}, $\mathcal{S}$ is not canonical, which is a contradiction.
		
		\begin{figure}[H]
			\centering
			\scalebox{.8}{	\between
			}
			\caption{\hyperref[between]{\text{Observation} \ref{between}}}\label{fig9}
		\end{figure}
		
	}
\end{proofing}

Now we can prove \hyperref[lem:dangerous]{\text{Observation} \ref{lem:dangerous}}.

\begin{proofing}{Lemma \ref{lem:dangerous}}
	{Let $\mathcal{S}$ be a canonical path partition. Let $P$ be a path in $\mathcal{S}$. Let $x_1v_3$ be two consecutive vertices in $P$ where $x_1$ is heavy and $v_3$ is dangerous. (The case where the order is $v_3x_1$, can be handled in a similar way by reversing the order of all vertices in $P$.) Let $(v_3,u)\in E_F$ be a free edge that is incident to a vertex $u$ which has a path neighbor in $V_2$ (hence $u\in V_4\cup V_2^b\cup V_3$).
		\hyperref[different]{\text{Observation} \ref{different}} implies that $u$ is on the same path as $v_3$. This proves item \ref{itm:1} in \hyperref[lem:dangerous]{\text{Lemma} \ref{lem:dangerous}}. Thereafter, the combination of \hyperref[samesides]{\text{Observation}  \ref{samesides}} and \hyperref[between]{\text{Observation} \ref{between}} exclude all cases except for the following: $u$ is on the right of $v_3$ and has only one path neighbor in $V_2$, where this path neighbor (that we denote by $x_2$) is on the left of $u$. Consequently, path $P$ can be represented as $P= o_1P_1x_1v_3y_1P_2ux_2P_3o_2$. This proves item \ref{itm:2} in \hyperref[lem:dangerous]{\text{Lemma} \ref{lem:dangerous}}.
		
		\begin{figure}[H]
			\centering
			\scalebox{.8}{	\oppositesides
			}
			\caption{\hyperref[lem:dangerous]{\text{Lemma} \ref{lem:dangerous}}}\label{fig10}
		\end{figure}
		
		The following proposition proves item \ref{itm:3} in \hyperref[lem:dangerous]{\text{Lemma} \ref{lem:dangerous}}.
		
		\begin{prop}\label{nine}
			There is only one balanced edge that is incident to $x_2$. The other end-vertex of this edge is $o_2$.
		\end{prop}
		
		\begin{proofing}{Proposition \ref{nine}}
			{
				Let $o_{x_2}\in V_1$ be a vertex that is incident to $x_2$ through a balanced edge. There are some cases to consider:
				\begin{enumerate}
					\item If $o_{x_2} = o_1$, then we create from $P$ the path $P'=o_2P_3x_2o_1P_1x_1v_{3}uP_2y_1$ in which $y_1$ is an end-vertex of $P'$. Because $y_1$ is moderate, then it has at least two balanced edges. Hence, by the Pigeonhole principle, $y_1$ must be adjacent either to a component $P_{y_1}\neq P$ (that can be a cycle) or to $o_2$.
					\begin{itemize}
						\item
						If $y_1$ is incident to a component $P_{y_1}$, then we can create one path out of the two component $P_{y_1}$ and $P'$, contradicting Property ~1.
						
						\item	If $y_1$ is incident to $o_2$, then $P'$ can be made into a cycle, contradicting Property ~2.
					\end{itemize}
					\item If $o_{x_2}$ lies on a cycle $C_1$, there are three cases we have to exclude. Recall that $y_1$ is moderate and therefore $y_1$ has at least one balanced edge that is incident to a vertex of a path. We denote this vertex by $o_{y_1}$ and the path by $P_{y_1}.$
					
					\begin{itemize}
						\item  If $o_{y_1} = o_2$, then we create from $P$ the path $o_1P_1x_1v_3uP_2y_1o_2P_3x_2C_1$, contradicting Property $1$.
						
						\item  If $o_{y_1}= o_1$, then we use the fact that $x_1$ is heavy and
						hence must have at least one edge to an external path $P_{x_1}$. We then create $o_2P_3x_2C_1$ and $P_{x_1}x_1P_1o_1y_1v_3uP_2$, contradicting Property $1$.
						
						\item If $o_{y_1}$ lies on a path $P_{y_1}\neq P$, then we can create $o_1P_1x_1v_3uP_2y_1P_{y_1}$ and $o_2P_3x_2C_1$, two components out of three, contradicting Property $1$.
					\end{itemize}
					\item If $o_{x_2}$ is an end-vertex of path $P_{x_2}\neq P$, then we create a path partition with one more cycle and the same number of components, contradicting Property $2$. The fact that $x_1$ is heavy and hence has three balanced edges to paths, implies that one of them must be incident to $o_{x_1}\neq o_{x_2},o_1$. The options are as follows:
					\begin{itemize}
						\item If $o_{x_1} = o_2$, then we create from $P$ and $P_{x_2}$ one path $o_1P_1x_1o_2P_3x_2P_{x_2}$ and one cycle $v_3y_1P_2u$, making a path partition with the same number of components and one more cycle, contradicting Property $2$.
						\item If $o_{x_1}$ lies on the path $P_{x_2}$ (but is not $o_{x_2}$), then by connecting $x_1$ to the end-vertex of $P_{x_2}$ we create from $P$ and $P_{x_2}$ one path $o_1P_1x_1P_{x_2}x_2P_3o_2$ and one cycle $v_3y_1P_2u$, contradicting Property $2$.
						\item If $o_{x_1}$ is on a path $P_{x_1}\neq P,P_{x_2}$ then we create from these three paths two paths $o_2P_3x_2P_{x_2}$, $P_{x_1}x_1P_1o_1$ and one cycle $v_3y_1P_2u$, contradicting Property $2$.
					\end{itemize}
				\end{enumerate}
				Hence we conclude that $o_{x_2} = o_2$.
			}
		\end{proofing}

		The following proposition proves item \ref{itm:4} in \hyperref[lem:dangerous]{\text{Lemma} \ref{lem:dangerous}}.
		\begin{prop}\label{nineone}
			Given that $x_2$ has an edge to $o_2$, then only one of $y_1's$ balanced edges is incident to a vertex of a path, and the other end-vertex of this edge is $o_2$.
		\end{prop}

		\begin{proofing}{Proposition \ref{nineone}}
			{ Let $o_{y_1}\in V_1$ be a path vertex that is incident to $y_2$ by a balanced edge. There are two cases to exclude:
				\begin{enumerate}
					\item  $x_1$ is heavy and hence must have at least one edge that is incident to an end-vertex of a path $P_{x_1}\neq P$. If $o_{y_1}=o_1$ then we can create from $P$ and $P_{x_1}$ the path $v_3uP_2y_1o_1P_1x_1P_{x_1}$ and the cycle $x_2P_3o_2$, contradicting Property $2$ \\
					\item If $o_{y_1}$ is an end-vertex of an external path $P_{y_1}$ then from $P$ and $P_{y_1}$  we create  $o_1P_1x_1v_3uP_2y_1P_{y_1}$ and the cycle $x_2P_3o_2$, contradicting Property $2$.
				\end{enumerate}
				Hence $o_{y_1} = o_2$.
			}
		\end{proofing}

	}
\end{proofing}

\subsection{Properties of $V_2$, and proofs of \hyperref[lem:kadjacent]{\text{Lemma} \ref{lem:kadjacent}} and  \hyperref[lem:special_case]{\text{Lemma} \ref{lem:special_case}}} \label{subsec:pro2}

\subsubsection{Proof of \hyperref[lem:kadjacent]{\text{Lemma} \ref{lem:kadjacent}}} \label{subsub:lemma4}

\hyperref[lem:kadjacent]{\text{Lemma} \ref{lem:kadjacent}} is a result of the following two lemmas.

\setcounter{lemma}{7}
\begin{lemma}\label{lem:none_to_dangerous}	
	Let $P=o_1P_1XP_2o_2$  be a path in a canonical path partition where $X = x_1x_2...x_k$ is a sequence of vertices in $V_2^b$. If $X$ has no free edges to dangerous vertices then it can transfer at most $(2+k)\cdot\frac{2}{3}$ points. Consequently, $X$  ends with at least $\frac{1}{3}\cdot k-\frac{4}{3}$ points.
	
\end{lemma}

\begin{lemma}\label{lem:has_to_dangerous}	
	Let $P=o_1P_1XP_2o_2$  be a path in a canonical path partition where $X = x_1x_2...x_k$ is a sequence of vertices in $V_2^b$.  If there is a vertex $x_i\in X$ that has an edge to a dangerous vertex then $X$ transfers at most $(k+1)\cdot\frac{2}{3} +\frac{5}{12}$ points. Consequently, $X$  ends with at least $\frac{1}{3}\cdot k-\frac{13}{12}$ points.
	
\end{lemma}

\setcounter{lemma}{8}

At first we refine our definitions for vertices in $V_2$.

\begin{itemize}
	
	\item We say that $x\in V_2$ \textbf{ goes} to $y$ if the edges $(x,y)\in E$ is a balanced edge.
	\begin{itemize}
		\item We say that $x\in V_2$ \textbf{goes} a path, if $x$ goes to $y\in V_1$ and $y$ is an end vertex of a path.
		\item We say that $x\in V_2$ \textbf{goes} a cycle, if $x$ goes to $y\in V_1$ and $y$ is lies on a cycle.
	\end{itemize}

	\item A vertex $y\in V_2$ in some path $P\in \mathcal{S}$  is an \text{\em inner} if it \textbf{goes} to one of the end-vertices of $P$.
	\item A pair of vertices $x_1,x_2\in V_2$ in some path $ P=o_1P_1x_1P_2x_2P_3o_2\in \mathcal{S}$ will be {\em crossing-inners} if $x_1$ \textbf{goes} to $o_2$ and $x_2$ \textbf{goes} to $o_1$.
	\item A pair of vertices $x_1,x_2\in V_2$ in some path $ P=o_1P_1x_1P_2x_2P_3o_2\in \mathcal{S}$ will be {\em splitting-inners} if $x_1$ \textbf{goes} to $o_1$ and $x_2$ \textbf{goes} to $o_2$.

\end{itemize}

Given a set $X \subset V_2$ we use the following notation:
\begin{itemize}
	
	\item $N_b(X)$ - The set of all vertices in $V_1$ that $X$ \textbf{goes} to.
	
	\item $b(X)$ - the total points that vertices of $X$ transfer by the transition rules using balanced edges (by \ref{rule:v2_cycle} and \ref{rule:v2_path}).
	
\end{itemize}

Hence, in this new notation \hyperref[lem:none_to_dangerous]{\text{Lemma }\ref{lem:none_to_dangerous}} states that $\displaystyle b(X)=\sum_{j=1}^{k}b(x_j)\leq \frac{2}{3}\cdot(2+k)$.
In order to prove the lemma we use some observations on path neighbor vertices in $V_2^b$. In the following observations we consider a path $P=o_1P_1x_1x_2P_2o_2$ in a canonical path partition, in which $x_1,x_2\in V_2^b$ are path neighbors. In some of the observations, we will prove things for $x_1$. These observations will imply the same results for $x_2$, but in a symmetric fashion.

\begin{obs}\label{crossinginacceptor}
	There are no path neighbors that are crossing-inners in a canonical path partition.	
\end{obs}
\begin{figure}[H]
	\centering
	\scalebox{.8}{\crossing
	}
	\caption{\hyperref[crossinginacceptor]{\text{Observation} \ref{crossinginacceptor}} - Crossing-inners}\label{fig12}
\end{figure}

\begin{proofing}{Observation \ref{crossinginacceptor}}
	{ Let $P=o_1P_1x_1x_2P_2o_2$ where $x_1,x_2$ are crossing-inners. Then we create from $P$ the cycle $o_1P_1x_1o_2P_2x_2$, contradicting the assumption that $P$ is a path and not a cycle.		
	}
\end{proofing}

\begin{obs}\label{inacceptor1}
	If $x_1$ \textbf{goes} to $o_{2}$, then $N_b(x_2)= \{o_2\}$.
\end{obs}

\begin{proofing}{Observation \ref{inacceptor1}}
	{
		Let $P=o_1P_1x_1x_2P_2o_2$ and $(x_1,o_2)\in E$. Let $o_{x_2}$ be a vertex that $x_2$ \textbf{goes} to. We will show that $o_{x_2} = o_2$. Here too, there are some cases to consider.
		\begin{enumerate}[label=\roman*]
			\item If $o_{x_2}= o_1$, then $x_1$ and $x_2$ are crossing-inners and this contradicts \hyperref[crossinginacceptor]{\text{Observation} \ref{crossinginacceptor}}.
			\item If $o_{x_2}$ lies on a component $P_{x_2}\neq P$, we will create $o_1P_1x_1o_2P_2x_2P_{x_2}$, making a path partition with less components, contradicting Property $1$.
		\end{enumerate}
		Hence $o_{x_2}= o_2$.
		\begin{figure}[H]
			\begin{subfigure}[t]{0.4\columnwidth} \centering \resizebox{\linewidth}{!}
				{	\crossing}
				\caption{$o_{x_2} = o_1$ } \label{fig13:subfig1}
			\end{subfigure}
			\hspace{20mm}
			\begin{subfigure}[t]{0.4\columnwidth} 	 \centering \resizebox{\linewidth}{!}
				{\inacceptor}
				\caption{$P_{x_2} \neq P$ } \label{fig13:subfig2}
			\end{subfigure}
			\caption{\hyperref[inacceptor1]{\text{Observation} \ref{inacceptor1}}}\label{fig13}
		\end{figure}
	}
\end{proofing}

\begin{obs}\label{inacceptor2}
	If $x_1$ \textbf{goes} to $o_{1}$, then $x_2$ can \textbf{go} either to end-vertices of $P$, to vertices of cycles, or to both.
\end{obs}

\begin{proofing}{Observation \ref{inacceptor2}}
	{
		Let $P=o_1P_1x_1x_2P_2o_2$, and suppose that $x_1$ \textbf{goes} to $o_1$. Assume in contradiction that $x_2$ \textbf{goes} to a vertex on an external path $P_{x_2}\neq P$. Then a path $o_2P_2x_2P_{x_2}$ and a cycle $x_1P_1o_1$ will be created, making a path partition with the same number of components and one more cycle, contradicting Property $2$.
	}
\end{proofing}

\begin{obs}\label{cycleout}
	If $x_1$ \textbf{goes} to a cycle $C_1$, then $x_2$ \textbf{goes} to $o_2$ or to a vertex on $C_1$ or to both.	
\end{obs}

\begin{proofing}{Observation \ref{cycleout}}
	{
		Let $P=o_1P_1x_1x_2P_2o_2$ and suppose that $x_1$ \textbf{goes} to the vertex $o_{x_1}$ which lies on a cycle $C_1$. Let $o_{x_2}$ be the vertex $x_2$ \textbf{goes} to. We consider the following cases:
		\begin{enumerate}
			\item If $o_{x_2} = o_1$, then applying  \hyperref[inacceptor1]{\text{Observation} \ref{inacceptor1}} on $x_2$ implies that $N_b(x_1)= \{o_2\}$, which is a contradiction.
			\item If $o_{x_2}$ is on $P_{x_2} \neq C_1 $, then we can create two paths $o_2P_2x_2P_{x_2}$ and $o_1P_1x_1C_1$  making a new path partition of size $|\mathcal{S}|-1$, contradiction to Property~1.
		\end{enumerate}
	}
\end{proofing}
\begin{obs}\label{pathout}
	If $x_1$ \textbf{goes} to a vertex $o_{x_1}$, which is an end-vertex of a path $P_{x_1}\neq P$, then $N_b(x_2)=\{o_{x_1}\}$.
\end{obs}

\begin{proofing}{Observation \ref{pathout}}
	{
		Let $P=o_1P_1x_1x_2P_2o_2$ and suppose that $x_1$ \textbf{goes}  to $o_{x_1}$, which is an end-vertex of the path $P_{x_1}\neq P$. Let $o_{x_2}$ be a vertex that $x_2$ \textbf{goes} to. We will show that $o_{x_1}=o_{x_2}$ by excluding all other cases.
		\begin{enumerate}
			\item If $o_{x_2}=o_2$, then we create a cycle $x_2P_2o_2$ and a path $o_1P_1x_1P_{x_1}$, making a new path partition of size  $|\mathcal{S}|$ with one more cycle (because $P_{x_1}$ is not a cycle), contradicting Property~2.
			\item If $o_{x_2}=o_1$ then applying  \hyperref[inacceptor1]{\text{Observation} \ref{inacceptor1}} on $x_2$ implies that $N_b(x_1)= \{o_2\}$, which is a contradiction.
			\item If $o_{x_2}$ is the other end-vertex of path $P_{x_1}$, then we concatenate $P$ and $P_{x_1}$ into one path  $o_1P_2x_2P_{x_1}x_1P_1o_1$, contradicting Property~1.
			\item If $o_{x_2}$ is on $P_{x_2} \neq P_{x_1}$ (it can be a cycle), we then create two paths $o_2P_2x_2P_{x_2}$,  $o_1P_1x_1P_{x_1}$, making a new path partition of size  $|\mathcal{S}|-1$, contradicting Property~1.
		\end{enumerate}
	}
\end{proofing}

Now we will generalize our observations for any $k\geq 2$ vertices.
In the following observations $P=o_1P_1XP_2o_2$ is a path in a canonical path partition where $X=x_1x_2...x_k$ is a sequence of $k\geq 2$ vertices in $V_2^b$.

\begin{obs}\label{pathoutk}
	If $x_i$ \textbf{goes} to a vertex $o_{x_i}$ that is an end-vertex of a path $P_{x_i}\neq P$, then $N_b(X)=\{o_{x_i}\}$. Consequently, $b(X)=k\cdot\frac{2}{3}$.
\end{obs}

\begin{proofing}{Observation \ref{pathoutk}} {
		For $o_{x_i}$ as above, \hyperref[pathout]{\text{Observation} \ref{pathout}} implies that the two path-neighbors of $x_i$ must \textbf{go} to $o_{x_i}$, which in turn implies that $x_i$ goes only to $o_{x_i}$. By induction we get that $N_b(X)=\{o_{x_i}\}$. Hence $b(X)=k\cdot\frac{2}{3}$.
	}
\end{proofing}

\begin{obs}\label{pathink}
	If $x_i$ \textbf{goes} to $o_2$ then $\forall j>i, N_b(x_j)=\{o_2\}$.
	If $x_i$ \textbf{goes} to $o_1$ then $\forall j<i, N_b(x_j)=\{o_1\}$.
	
\end{obs}

\begin{proofing}{Observation \ref{pathink}} {
		If $x_i$ goes to $o_2$ then from \hyperref[inacceptor1]{\text{Observation} \ref{inacceptor1}} $N_b(x_{i+1})=\{o_2\}$ and by induction $\forall j>i, N_b(x_j)=\{o_2\}$.
		Likewise, if $x_i$ \textbf{goes} to $o_1$, then $\forall j<i, N_b(x_j)=\{o_1\}$.
		
	}
\end{proofing}

\begin{obs}\label{cycleoutk}
	If $x_i$ \textbf{goes} to a cycle, then:
	\begin{enumerate}
		\item $\forall j>i$, $x_j$ \textbf{goes} to the same cycle as $x_i$, or \textbf{goes} to $o_2$, or \textbf{goes} to both.
		\item $\forall j<i$, $x_j$ \textbf{goes} to the same cycle as $x_i$, or \textbf{goes} to $o_1$, or \textbf{goes} to both.
	\end{enumerate}
\end{obs}

\begin{proofing}{Observation \ref{cycleoutk}} {
		Let $o_{x_i}$ be a vertex on a cycle that $x_i$ \textbf{goes} to. \hyperref[cycleout]{\text{Observation} \ref{cycleout}} $x_{i+1}$ implies then that $x_{i+1}$ must \textbf{go} to the same cycle, or \textbf{go} to $o_2$. By induction, using either  \hyperref[cycleout]{\text{Observation} \ref{cycleout}} or \hyperref[inacceptor1]{\text{Observation} \ref{inacceptor1}} we get that $\forall j>i$, $x_j$ must \textbf{go} either to the same cycle, or to $o_2$.
		A similar proof can be given to $x_j$ with $j<i$.	
	}
\end{proofing}

Let $X'\subseteq X$ be a sequence of adjacent nodes in $X$ that all \textbf{go} to the same cycle $C$.
The following two observations show that on average, the number of points each vertex in $X'$ transfers to $C$ is at most $\frac{1}{2}$.

\begin{obs}\label{samecycle2}
	If $X'= x_1, x_2$  \textbf{go} to the same cycle $C$, then they cannot \textbf{go} to neighboring vertices along the cycle. Moreover, $X'$ transfer to $C$ (by \ref{rule:v2_cycle}) at most one point ($2\cdot\frac{1}{2}$).
\end{obs}

\begin{proofing}{Observation \ref{samecycle2}}
	{
		Assume that $X'= x_1, x_2$ \textbf{go} to vertices $o_{x_1}$ and $o_{x_2}$, that lie on the same cycle, $C$.
		If $(o_{x_2},o_{x_1})\in E$ is an edge that is part of the cycle, we open the cycle into a path with end-vertices $o_{x_1}$ and $o_{x_2}$. Thereafter, we crate the path  $o_2P_2x_2Cx_1P_1o_1$ out of $P$ and  $C$. Consequently, we have a new path partition of size $|\mathcal{S}|-1$, contradicting Property~1.

		Let $c$ be the size of cycle $C$.
		\begin{enumerate}
			\item If $c=3$ then only one vertex in $C$ can have an edges incident to $X'$, because any two vertices on $C$ are neighbors.
			Hence the number of edges that can transfer points from $X'$ to $C$ is at most $2$.  Consequently the number of points transferred to $C$ is at most $\frac{2}{3}$.		
			\item  If $c=4$ then at most two vertices in $C$ can have edges incident to $X'$. 	Hence the number of edges that can transfer points from $X'$ to $C$ is at most $4$.  Consequently the number of points transferred to $C$ is at most $1$.	
			

			\item  If $c=5$ then at most three vertices in $C$ have edges incident to $X'$, and the number of edges that can transfer points from $X'$ to $C$ is at most $4$.  Consequently the number of points transferred to $C$ is at most $\frac{4}{5}$.

			
			\begin{figure}[H]
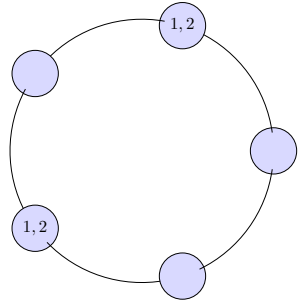
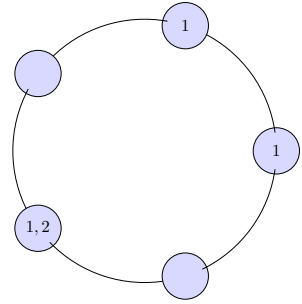

				\centering
				\begin{subfigure}[t]{0.3\columnwidth}  \centering \resizebox{0.8\linewidth}{!}
					{\fivecycletwo}
					\caption{\centering two vertices in $C$ have edges incident to $X'$ } \label{fig100:subfig1}
				\end{subfigure}
				\hspace{10mm}
				\begin{subfigure}[t]{0.3\columnwidth}
					\centering \resizebox{0.8\linewidth}{!}
					{\fivecyclethree}
					\caption{\centering three vertices in $C$ have edges incident to $X'$} \label{fig100:subfig2}
				\end{subfigure}
				\caption{The two cases of $c=5$ }\label{fig100}
			\end{figure}
			
			\item  If $c=6$ then at most four vertices in $C$ have edges incident to $X'$, and the number of edges that can transfer points from $X'$ to $C$ is at most $6$.  Consequently the number of points transferred to $C$ is at most $1$.

			
			\begin{figure}[H]
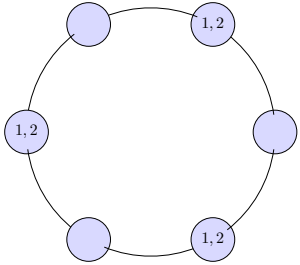
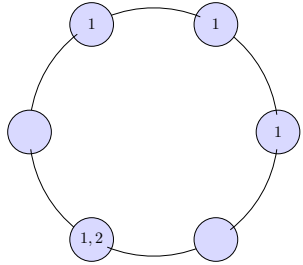

				\centering
				\begin{subfigure}[t]{0.3\columnwidth}  \centering \resizebox{0.8\linewidth}{!}
					{\sixcyclethree}
					\caption{\centering three vertices in $C$ have edges incident to $X'$ } \label{fig201:subfig1}
				\end{subfigure}
				\hspace{10mm}
				\begin{subfigure}[t]{0.3\columnwidth}
					\centering \resizebox{0.8\linewidth}{!}
					{\sixcyclefour}
					\caption{\centering four vertices in $C$ have edges incident to $X'$ } \label{fig201:subfig2}
				\end{subfigure}
				\caption{The two cases of $c=6$ }\label{fig201}
			\end{figure}

		\end{enumerate}

	}
\end{proofing}

\begin{obs}\label{samecycle3}
	If $X'= x_1,x_2,x_3$ in $X$ \textbf{go} to the same cycle $C$, then  $X'$ transfers to $C$ (by \ref{rule:v2_cycle}) at most $\frac{3}{2}$ points ($3\cdot\frac{1}{2}$).
	
\end{obs}

\begin{proofing}{Observation \ref{samecycle3}}
	{
		Suppose that $X'= x_1,x_2,x_3$ \textbf{go} to the same cycle $C$.
		Let $c$ be the size of cycle $C$. By \hyperref[samecycle2]{\text{Observation} \ref{samecycle2}} any two adjacent vertices among $x_1,x_2,x_3$ cannot \textbf{go} to neighboring vertices along the cycle.
		\begin{enumerate}
			\item If $c=3$ then only one vertex in $C$ can have an edges incident to $X'$, because any two vertices on $C$ are neighbors.			 Hence the number of edges that can transfer points from $X'$ to $C$ is at most $3$. Consequently, the number of points transferred to $C$ is at most $1$.	
			
			
			\item  If $c=4$ then at most two vertices in $C$ can have edges incident to $X'$.  Hence the number of edges that can transfer points from $X'$ to $C$ is at most $6$. Consequently, the number of points transferred to $C$ is at most $\frac{6}{4}$.	
			
			\item  If $c=5$ then at most three vertices in $C$ have edges incident to $X'$. Hence the number of edges that can transfer points from $X'$ to $C$ is at most $7$. Consequently, the number of points transferred to $C$ is at most $\frac{7}{5}$.

			
			\begin{figure}[H]
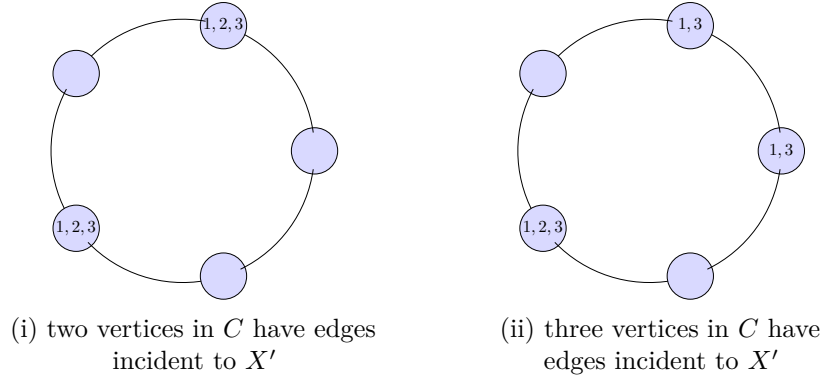

				\centering
				\begin{subfigure}[t]{0.3\columnwidth}  \centering \resizebox{0.8\linewidth}{!}
					{\fivecycletwob}
					\caption{\centering two vertices in $C$ have edges incident to $X'$ } \label{fig101:subfig1}
				\end{subfigure}
				\hspace{10mm}
				\begin{subfigure}[t]{0.3\columnwidth}
					\centering \resizebox{0.8\linewidth}{!}
					{\fivecyclethreeb}
					\caption{\centering three vertices in $C$ have edges incident to $X'$} \label{fig101subfig2}
				\end{subfigure}
				
				\caption{The two cases of $c=5$ }\label{fig101}
			\end{figure}

			\item  If $c=6$ then at most four vertex in $C$ have edges incident to $X'$. Therefore, the number of edges that can transfer points from $X'$ to $C$ is at most $9$.  Consequently, the number of points transferred to $C$ is at most $\frac{9}{6}$.


			\begin{figure}[H]
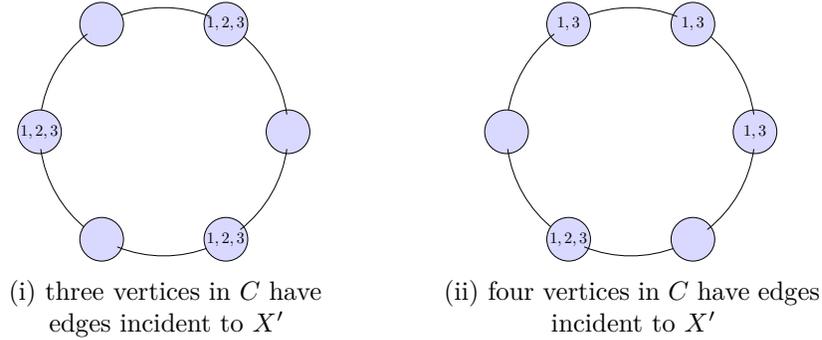

				\centering
				\begin{subfigure}[t]{0.3\columnwidth}  \centering \resizebox{0.8\linewidth}{!}
					{\sixcyclethreeb}
					\caption{\centering three vertices in $C$ have edges incident to $X'$ } \label{fig301:subfig1}
				\end{subfigure}
				\hspace{10mm}
				\begin{subfigure}[t]{0.3\columnwidth}
					\centering \resizebox{0.8\linewidth}{!}
					{\sixcyclefourb}
					\caption{\centering four vertices in $C$ have edges incident to $X'$} \label{fig301:subfig2}
				\end{subfigure}
				\caption{The two cases of $c=6$ }\label{fig301}
			\end{figure}

		\end{enumerate}

	}
\end{proofing}

\begin{obs}\label{samecyclek}
	Let $X'\subseteq X$ be all the vertices in $X$ that \textbf{go} to vertices on cycles. If $|X'|>1$, then all vertices in $X'$ are adjacent and \textbf{go} to the same cycle, and the the number of points $X'$ transfers to $C$ is at most $\frac{1}{2}\cdot|X'|$. Consequently, $b(X')\leq |X'|\cdot\frac{1}{2}+2\cdot\frac{2}{3}$.
\end{obs}

\begin{proofing}{Observation \ref{samecyclek}} {
		Consider $X'\subseteq X$ to be all the vertices in $X$ that \textbf{go} to vertices on cycles. We prove, that if $|X'|=2$, the vertices of $X'$ will be adjacent and \textbf {go} to the same cycle.	 Let $x_i$ and $x_j$ be the two vertices that make $X'$ where $i<j$. Assume towards contradiction that $j\neq i+1$, hence there is $i<l<j$ ,$x_l$ that does not \textbf{go} to any cycles. Applying  \hyperref[cycleoutk]{\text{Observation} \ref{cycleoutk}} on $x_i$ we get that $x_l$ must \textbf{go} to $o_2$, and only to it. Applying  \hyperref[cycleoutk]{\text{Observation} \ref{cycleoutk}} on $x_j$ implies that $x_l$ must go to $o_1$, which is contradiction.
		Hence $j=i+1$ and moreover both vertices \textbf{go} to the same cycle.  In the same way, a simple induction will prove that for any size of $|X'|\geq 2$, the vertices of $X'$ are adjacent and all have edges to the same cycle.
		From \hyperref[samecycle2]{\text{Observation} \ref{samecycle2}} and  \hyperref[samecycle3]{\text{Observation} \ref{samecycle3}} we get that $X'$ transfers at most $\frac{1}{2}\cdot|X'|$
		points to $C$.
		
		Let $x_i$ be the first vertex from left in $X'$ and $x_j$ be the last vertex from left. $x_i$ can \textbf{go} to $o_1$ and $x_j$ can \textbf{go} to $o_2$, in addition to \textbf{going} to $C$. Hence in total $b(X')\leq |X'|\cdot\frac{1}{2}+2\cdot\frac{2}{3}$.
		

	}
\end{proofing}

Now we prove \hyperref[lem:none_to_dangerous]{\text{Lemma} \ref{lem:none_to_dangerous}}.

\begin{proofing}{Lemma \ref{lem:none_to_dangerous}}
	{	
		
		Let $P=o_1P_1XP_2o_2$  be a path in a canonical path partition where $X = x_1x_2...x_k$ is a sequence of vertices in $V_2^b$. According to this lemma, $X$ has no free edges incident to dangerous vertices. There are four cases we need to consider (see Figure \ref{fig16}):
		\begin{enumerate}
			\item If there is $i$ such that $x_i$ \textbf{goes} to a path $P_{x_i}\neq P$, then from \hyperref[pathoutk]{\text{Observation} \ref{pathoutk}} we get that $b(X) = k\cdot\frac{2}{3}$. (See Figure \ref{fig16:subfig1}.)
			
			\item If all vertices in $X$ \textbf{go} to $o_1$ or $o_2$ then we will show that there can be only one vertex in $X$ that has two balanced edges.
			Assume in contradiction, that there are two vertices, $x_i$ and $x_j$ in which each one of them \textbf{goes} to $o_1$ and $o_2$. W.l.o.g suppose that $i<j$. Hence
			by \hyperref[pathink]{\text{Observation} \ref{pathink}} and the fact that $x_i$ \textbf{goes} to $o_2$ then $x_j$ cannot \textbf{go} to $o_1$ - contradiction. Hence $b(X)= (k+1)\cdot\frac{2}{3}$. (See Figure \ref{fig16:subfig2}.)\label{item2:lemma}

			\item If there is exactly one vertex $x_i\in X$ that \textbf{goes} to a cycle. (See Figure \ref{fig16:subfig3}.) From \hyperref[cycleoutk]{\text{Observation} \ref{cycleoutk}} we know that all $j<i$ must \textbf{go} to $o_1$ and all $j>i$ must \textbf{go} to $o_2$. Moreover, $x_i$ can \textbf{go} to $o_1$, $o_2$ as well and to two cycles of size three.  Hence $b(X) = b(x_i)+ (k-1)\cdot\frac{2}{3}\leq  2\cdot\frac{2}{3}+2\cdot\frac{1}{3} +(k-1)\cdot\frac{2}{3} = (k+2)\cdot\frac{2}{3}$. \label{item3:lemma}

			\item Let $X'\subseteq X$ be the set of all the vertices in $X$ that \textbf{go} to cycles. (See Figure \ref{fig16:subfig4}.) If  $|X'|\geq 2$ then from \hyperref[samecyclek]{\text{Observation} \ref{samecyclek}} it must be that all vertices in $X'$  are adjacent and $b(X')\leq|X'|\cdot\frac{1}{2}+2\cdot\frac{2}{3}$. Let $x_i$ be the first in $X'$ and $x_j$ be the last one.
			Then from  \hyperref[cycleoutk]{\text{Observation} \ref{cycleoutk}} for all $l<i$ $x_l$, must \textbf{go} to $o_1$ and for all $l>j$ $x_l$, must go to $o_2$.
			Hence each vertex in this group will transfer $\frac{2}{3}$ points.   Hence we can conclude that $b(X) =(k-|X'|)\cdot\frac{2}{3}+ |X'|\cdot\frac{1}{2}+2\cdot\frac{2}{3} = (k+2)\cdot\frac{2}{3} -\frac{1}{6}\cdot|X'|$.	$|X'|\geq 2$ and hence $b(X) \leq (k+2)\cdot\frac{2}{3} -\frac{1}{3}$.\label{item4:lemma}


			

		\end{enumerate}

		Hence we can conclude that in all cases $b(X)\leq (k+2)\cdot\frac{2}{3}$. Consequently, $X$  ends with at least $\frac{1}{3}\cdot k-\frac{4}{3}$ points.
		
		\begin{figure}[H]
			\centering
			\begin{subfigure}[t]{0.45\columnwidth} \centering \resizebox{\linewidth}{!}
				{\kouterpath}
				\caption{$x_i$ goes to an external path } \label{fig16:subfig1}
			\end{subfigure}
			\hspace{5mm}
			\begin{subfigure}[t]{0.45\columnwidth} \centering \resizebox{\linewidth}{!}
				{\knonouter}
				\caption{All vertices in $X$ go to $o_1$ or $o_2$} \label{fig16:subfig2}
			\end{subfigure}
			\hspace{5mm}
			\begin{subfigure}[t]{0.45\columnwidth} \centering	  \resizebox{\linewidth}{!}
				{\ksequence}
				\caption{only $x_i$ goes to a vertex on a cycle} \label{fig16:subfig3}
			\end{subfigure}
			\hspace{5mm}
			\begin{subfigure}[t]{0.45\columnwidth} 	\centering  \resizebox{\linewidth}{!}
				{\koutercycles}
				\caption{At least two nodes in $X$ go to cycles} \label{fig16:subfig4}
			\end{subfigure}
			\caption{\hyperref[lem:none_to_dangerous]{\text{Lemma} \ref{lem:none_to_dangerous}}} \label{fig16}
		\end{figure}

	}
\end{proofing}



To prove \hyperref[lem:has_to_dangerous]{\text{Lemma} \ref{lem:has_to_dangerous}} we will need some observations.
\refstepcounter{lemma}

\begin{obs}\label{onlyone}
	Let $X$ be $k$ adjacent vertices $x_1,x_2....x_k \in V_2^b$ on a path $P=o_1P_1x_1x_2...x_kP_2o_2$. If there is $x_i$ that has an edge to a dangerous vertex, then $i\in \{1,k\}$. Moreover,
	\begin{itemize}
		
		\item 	if $x_1$ has an edge to a dangerous vertex, then $N_b(X \setminus \{x_1\})= \{o_2\}$.
		
		\item 	if $x_k$ has an edge to a dangerous vertex, then $N_b(X \setminus \{x_k\})= \{o_1\}$.
		
	\end{itemize}
\end{obs}	

\begin{proofing}{Observation \ref{onlyone}}
	{Let $X$ be $k$ adjacent vertices $x_1,x_2....x_k \in V_2^b$ on a path $P=o_1P_1x_1x_2...x_kP_2o_2$. If $x_i$ has an edge that is incident to a dangerous vertex $v_3$, then  \hyperref[lem:dangerous]{\text{Lemma} \ref{lem:dangerous}} implies that $v_3$ is on $P$ and that $x_i$ cannot have  both its neighbors in $V_2$. Hence $i\in \{1,k\}$. W.l.o.g assume that $x_1$ has an edge to dangerous vertex. \hyperref[lem:dangerous]{\text{Lemma} \ref{lem:dangerous}} implies that $N_b(x_2)= \{o_2\}$.
		Using the fact that $x_3\in V_2$ we get from \hyperref[inacceptor1]{\text{Observation} \ref{inacceptor1}} that $N_b(x_3)= \{o_2\}$. With an easy induction we get that $N_b(X \setminus \{x_1\})= \{o_2\}$. The same proof will work for $x_k$.
		
	}
	
\end{proofing}

\begin{obs}\label{only}
	Let $X$ be two adjacent vertices $x_1,x_2\in V_2^b$ on a path $P=o_1P_1x_1x_2P_2o_2$.
	If both of them have an edge to a dangerous vertex, then $N_b(x_2)= \{o_2\}$ and $N_b(x_1)= \{o_1\}$. Consequently, $X$ transfers at most $2\cdot\frac{2}{3}+\frac{1}{2}$ points.

\end{obs}	

\begin{proofing}{Observation \ref{only}}
	{	Let $X$ be two adjacent vertices $x_1,x_2\in V_2^b$ on a path $P=o_1P_1x_1x_2P_2o_2$. If both of them have an edge to a dangerous vertex, then $N_b(x_2)= \{o_2\}$ and $N_b(x_1)= \{o_1\}$. (This follows from \hyperref[lem:dangerous]{\text{Lemma} \ref{lem:dangerous}}.) According to $\ref{rule:v4_dangerous}$ an edge that is not balanced can transfer at most $\frac{1}{12}$ points.
		
		Hence, in the case above, $X$ will transfer at most $2\cdot(\frac{2}{3}+3\cdot\frac{1}{12})=2\cdot\frac{2}{3}+\frac{1}{2}$ points.
		
	}
\end{proofing}

\begin{proofing}{Lemma \ref{lem:has_to_dangerous}} {	Let $\mathcal{S}$ be a canonical path partition.  Let $X$ be $k\geq 2$ adjacent vertices $x_1,x_2....x_k \in V_2^b$ on a path $P=o_1P_1x_1x_2...x_kP_2o_2$. Suppose that there is $i$ such that $x_i$ has an edge to a dangerous vertex. We prove that $X$ transfers at most $(k+1)\cdot\frac{2}{3} +\frac{5}{12}$ points. We break the proof into two cases:
		
		\begin{enumerate}
			\item If there is exactly one $x_i$ which has an edge to a dangerous vertex, then by \hyperref[onlyone]{\text{Observation} \ref{onlyone}} we can assume w.l.o.g that $x_1$ is the only vertex that has an edge to a dangerous vertex and $N_b(X \setminus \{x_1\})= \{o_2\}$. Hence for every $2\leq i\leq k$, $x_i$ will transfer $\frac{2}{3}$ points. In particular, $x_2$ \textbf{goes} to $o_2$ and hence \hyperref[pathout]{\text{Observation} \ref{pathout}} implies $x_1$ cannot \textbf{go} to an end-vertex of an external path. Hence, $x_1$ can have its remaining three free edges either to cycles, to ends vertices of $P$ ($o_1$ or $o_2$) or to dangerous vertices. If $x_1$ has an edge to $o_2$, $o_1$ and to a cycle of size three it will transfer  $\frac{1}{12}+2\cdot\frac{2}{3}+\frac{1}{3}$ points. Any other case will lead $x_1$ to transfer less points.
			Hence we get that $X$ will transfer at most
			$\frac{1}{12}+2\cdot\frac{2}{3}+\frac{1}{3}+(k-1)\cdot\frac{2}{3}$
			$=(k+1)\cdot\frac{2}{3} +\frac{1}{3}+\frac{1}{12}=(k+1)\cdot\frac{2}{3} +\frac{5}{12}$ points.
			
			
			\item If there are at least two vertices in $X$ that have an edge to dangerous vertices, then we show that it must be that $k=2$. Assume in contradiction that $k>2$.
			By \hyperref[onlyone]{\text{Observation} \ref{onlyone}} and because $k>2$, $x_2\neq x_k$ and $N_b(x_2)= \{o_1\}$ and $N_b(x_2)= \{o_2\}$ leading to a contradiction. Hence
			we can assume that $k=2$ and both vertices have edges to dangerous vertices. By \hyperref[only]{\text{Observation} \ref{only}}, $X$ transfers at most $2\cdot\frac{2}{3} +\frac{1}{2}$ points.
			
			
		\end{enumerate}

		Hence we can conclude that in all cases, $X$ transfers at most $(k+1)\cdot\frac{2}{3} +\frac{5}{12}$ points. Consequently, $X$  ends with at least $\frac{1}{3}\cdot k-\frac{13}{12}$ points.
		
	}
\end{proofing}

\subsubsection{Proof of \hyperref[lem:special_case]{\text{Lemma} \ref{lem:special_case}}} \label{subsub:lemma5}
\refstepcounter{lemma}
\setcounter{lemma}{5}
We prove now the following observations for the specific case we need to handle.

\begin{obs}\label{onlyonetocycles}
	Let $P=o_1P_1x_1x_2P_2o_2$ be a path, where $x_1$ and $x_2$ are in $V_2$ and have no free edges to dangerous vertices. If  $x_1$ and $x_2$ transfer more than $\frac{7}{3}$ points ($b(x_2)+b(x_1) >\frac{7}{3}$), then $x_2$ and $x_1$ are splitting-inners.

\end{obs}

\begin{proofing}{Observation \ref{onlyonetocycles}}
	{ Let $P=o_1P_1x_1x_2P_2o_2$ where $x_1$ and $x_2$ are in $V_2$ and have no free edges to dangerous vertices. According to the premises of the lemma $b(x_2)+b(x_1)> \frac{7}{3}$. The proof of \hyperref[lem:none_to_dangerous]{\text{Lemma} \ref{lem:none_to_dangerous}} implies that we must be in case similar to case  \ref{item3:lemma} of the proof, because in all other cases $b(x_1)+b(x_2)\leq  \frac{7}{3}$. Hence we are in the case where either $x_1$ or $x_2$ go to a vertex on a cycle. W.l.o.g we will assume that it is $x_1$. \hyperref[cycleout]{\text{Observation} \ref{cycleout}} implies that $x_2$ must go to $o_2$, and only to $o_2$. Hence $b(x_2)=\frac{2}{3}$.
		Consequently, $b(x_1) > \frac{5}{3}$. Since any balanced edge to a cycle transfers at most $\frac{1}{3}$ points, $x_1$ must have at least two balanced edges to paths. Moreover, $x_1$ cannot go to vertices that lie on other paths, by \hyperref[inacceptor2]{\text{Observation} \ref{inacceptor2}}. Hence $x_1$ must have an edge to $o_1$, making $x_1$ and $x_2$ splitting-inners.
	}
\end{proofing}

\begin{obs}\label{splittinginacceptors}
	Let $P=o_1P_1x_1x_2P_2o_2$ be a path. If $x_1$ and $x_2$ are adjacent splitting-inners, then no vertex $u$ on $ P_1$ or on $P_2$ is heavy.
\end{obs}

\begin{proofing} {Observation \ref{splittinginacceptors}}{
		Let $P=o_1P_1x_1x_2P_2o_2$ be a path where $x_1$ and $x_2$ are splitting inners. Assume towards contradiction that $P_1= P_1^{'}uP_1^{''}$ where $u$ is heavy. Therefore, $u$ must go to a vertex on a path $P_{u}\neq P$. We will create a cycle component $x_2P_2o_2$, and the path $P_1^{''}x_1o_1P_1^{'}uP_u$, making a path partition with the same number of components and one more cycle, contradicting Property $2$.
		\begin{figure}[H]
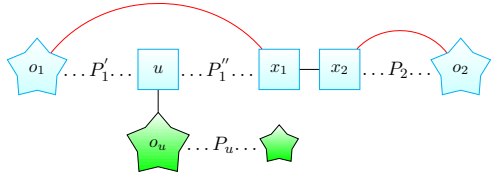

			\centering
			\resizebox{0.4\textwidth}{!}{%
				
				\splittinginacceptors
				
			}
			\caption{splitting-inners}
		\end{figure}
		
	}
\end{proofing}

\begin{proofing}{Lemma \ref{lem:special_case}}
	{	Let $\mathcal{S}$ be a canonical path partition and let $P$ be a path in $\mathcal{S}$. Let $B_iB_{i+1}$ be two consecutive blocks in path $P$ such that $B_i$ is of kind $1$ ($X_i=x_i$) and $|X_{i+1}|=2$ ($X_{i+1}=x_{i+1}^1x_{i+1}^2$). 		Suppose that $x_i$ is heavy. There are three cases we need to handle:
		\begin{enumerate}
			\item If all free edges of $X_{i+1}$ are balanced edges then we show that $b(X_{i+1})\leq \frac{7}{3}$ and hence transfers at most  $\frac{7}{3}$ points. Assume in contradiction that $b(X_{i+1})>\frac{7}{3}$. Hence, by  \hyperref[onlyonetocycles]{\text{Observation} \ref{onlyonetocycles}}, $x_{i+1}^1$ and $x_{i+1}^2$ are splitting-inners. Therefore, from \hyperref[splittinginacceptors]{\text{Observation} \ref{splittinginacceptors}} $x_i$ is not heavy, contradiction.
			\item If both $x_{i+1}^1$ and $x_{i+1}^2$ have an edge to a dangerous vertex, then by \hyperref[only]{\text{Observation} \ref{only}} they transfer at most $2\cdot\frac{2}{3}+\frac{1}{2}$ points.
			\item If there is only one vertex that has an edge to a dangerous vertex, then w.l.o.g assume it is $x_{i+1}^1$. (The same proof will apply for $x_{i+1}^2$.)
			
			Let $o_1$ be the left end-vertex of path $P$ and $o_2$ be the right one. 		 \hyperref[onlyone]{\text{Observation} \ref{onlyone}} implies that $N_b(x_{i+1}^2)= \{o_2\}$. As a result, there is no other edge through which $x_{i+1}^2$ can transfer points, and hence $x_{i+1}^2$ will transfer $\frac{2}{3}$ points. On the other hand, $x_{i+1}^1$ cannot have an edge to $o_1$ because then they will be splitting inners, contradicting the fact that $x_i$ is heavy (\hyperref[splittinginacceptors]{\text{Observation} \ref{splittinginacceptors}}).
			\hyperref[inacceptor2]{\text{Observation} \ref{inacceptor2}} implies that
			the remaining three balanced edges incident to $x_{i+1}^1$  are incident to cycles, to $o_2$ or to dangerous vertices. On an edge to $o_1$, $x_{i+1}^1$ transfers $\frac{2}{3}$ points by \ref{rule:v2_path}. A balanced edge to a cycle will transfer at most $\frac{1}{3}$ points by \ref{rule:v2_cycle}. A free edge that is not balanced can transfer at most $\frac{1}{12}$ points by \ref{rule:v4_dangerous}.
			Hence the worst case is if the three edges incident to $x_{i+1}^1$ are to $o_2$ and to two cycles of size three.
			Hence in total $x_{i+1}^1$ will transfer at most $\frac{1}{12}+\frac{2}{3}+2\cdot\frac{1}{3}=2\cdot\frac{2}{3}+\frac{1}{12}$ points.
			
			Hence $X_{i+1}$ will transfer at most $2\cdot\frac{2}{3}+\frac{1}{12}+\frac{2}{3}=3\cdot\frac{2}{3}+\frac{1}{12}$  points.
			
		\end{enumerate}
		Hence in all cases, $X_{i+1}$ will transfer at most $\frac{7}{3}$ points. Consequently, as $X_{i+1}$ starts with two points, it has at least $-\frac{1}{3}$ points.

	}
\end{proofing}

\section{Proof of \hyperref[thm:main5]{\text{Theorem} \ref{thm:main5}} }

To prove \hyperref[thm:main5]{\text{Theorem} \ref{thm:main5}}
we modify the point transfer rules in the following way: \ref{rule:v2_path} and \ref{rule:v5_rest} remain unchanged,
\ref{rule:v4_dangerous} and \ref{rule:v2a_dangerous} are deleted, and \ref{rule:v2_cycle} is changed such that the points that are transferred to cycles on each edge is multiplied by $ \frac{4}{3}$ (hence $\frac{1}{i}\cdot\frac{4}{3}$). \label{rule:v2_cycle_changed}

For completeness we provide here the changed point transfer rules:

\begin{mdframed}[linecolor=black!40,
	outerlinewidth=1pt,
	roundcorner=.5em,
	innertopmargin=1.3ex,
	innerbottommargin=.5\baselineskip,
	innerrightmargin=1em,
	innerleftmargin=0.4em,
	backgroundcolor=blue!10,
	shadow=true,
	shadowsize=6,
	shadowcolor=black!20,
	frametitle={\Large Changed Transfer rules:},
	frametitlebackgroundcolor=cyan!40,
	frametitlerulewidth=14pt
	]

	\begin{description}[leftmargin=100pt
		]
		\item[\namedlabel{changed_rule:v2_cycle}{Changed Rule 1}]  From $V_2$ to $V_1\cap \mathcal{C}$. $v\in V_2$ transfers $\frac{1}{i}\cdot\frac{4}{3}$ points to $u\in V_1\cap \mathcal{C}$ , if  $u$ is a vertex of a cycle of size $i \le 6$ and $(u,v)$ is a balanced edge.
		
		\item [\namedlabel{changed_rule:v2_path}{Changed Rule 2}] From $V_2$ to $V_1\cap \mathcal{P}$. $v\in V_2$ transfers $\frac{2}{3}$ points to $u\in V_1$ if $u$ is an end-vertex of a path and $(u,v)$ is a balanced edge.

		\item [\namedlabel{changed_rule:v5_rest}{Changed Rule 5}] From $V_5$ to $V \setminus V_5$. $v\in V_5$ transfers $\frac{1}{4}$ points to $u\in (V \setminus V_5)$ if $(u,v)$ is a free edge.
		
	\end{description}
\end{mdframed}

This change gives us the following propositions:

\begin{pro}
	\label{pro:easy6}
	Let $C$ be a cycle in a canonical path partition. Then after applying the transfer rules $C$ has at least~$6+\frac{1}{3}$ points.
\end{pro}

\begin{proofing}{Proposition \ref{pro:easy6}}
	{
		Let $i$ denote the size of the cycle. If $i \ge 7$ then we are done, because according to the point transfer rules, a vertex of a cycle can only receive points. If $i=6$, then there are at least two balanced edges incident to the cycle, as there are no $K_6$ in the graph. Hence cycles of size $6$ ends with at least $6+\frac{4}{9}$ points. If $i \leq 5$ (and necessarily $i \ge 3$), then each vertex has at least $6-i$ balanced edges connecting it to vertices in $V_2$, because the degree of each vertex is~5. Each balanced edge contributes to the vertex $\frac{1}{i}\cdot \frac{4}{3}$ points (by \ref{changed_rule:v2_cycle}), hence, each vertex on the cycle has at least $1 + (6-i)\frac{1}{i}\cdot \frac{4}{3}$ points. Hence the cycle has at least $i\cdot(1 + (6-i)\frac{1}{i}\cdot \frac{4}{3}) = 8-\frac{i}{3}$ points. Hence, $i\leq 5$ there are at least $6+\frac{1}{3}$ points as desired.
		
	}
\end{proofing}

\begin{pro}
	\label{pro:easyouter6}
	Let $P$ be a path in a canonical path partition. Then after applying the transfer rules the end-vertices of $P$ have $6+ \frac{4}{3}$ points.
\end{pro}

\begin{proofing}{Proposition \ref{pro:easyouter6}}
	{In a canonical path partition, there are no isolated vertices. Hence $P$ has~8 balanced edges incident with its two end-vertices. Applying \ref{changed_rule:v2_path} together with the two starting points, sums up to $2 + 8\cdot \frac{2}{3}=6+ \frac{4}{3}$ points. }
\end{proofing}

\begin{pro}\label{pro:easypoints5}
	The only nodes that can end up with a negative number of points are those from $V_2$. In particular:
	\begin{enumerate}
		\item Every node in $V_2$ has at least $-1$ points.
		\item Every node in $V_4$ or $V_3$ has at least one point.
		\item Every node in $V_5$ has a non-negative number of points.

	\end{enumerate}
\end{pro}

\begin{proofing}{Proposition \ref{pro:easypoints5}}
	{	Each vertex starts with one point.
		\begin{enumerate}
			\item	Vertices in $V_2$ have at most three balanced edges, and may transfer at most $\frac{2}{3}$ points on each balanced edge, hence they have at least $-1$ points.
			\item	Vertices in $V_3\cup V_4$ can only receive points, hence they have at least one point.
			\item Vertices in $V_5$ have three free edges, and may transfer at most $\frac{1}{4}$ points on each free edge.
		\end{enumerate}
	}
\end{proofing}

For completeness we state the following proposition, that is identical to \hyperref[lem:kadjacent]{\textbf{Lemma} \ref{lem:kadjacent}} .

\begin{pro}
	\label{pro:lemmafour}
	Let $X$ be a sequence of $k\geq 2$ vertices in $V_2^b$, then $X$ ends with at least $\frac{1}{3}\cdot k-\frac{4}{3}$ points.
\end{pro}

\begin{proofing}{Proposition \ref{pro:lemmafour}}
	{

		We prove this, mimicking the proof of \hyperref[lem:kadjacent]{\textbf{Lemma} \ref{lem:kadjacent}} with a change in the following places:
		
		\begin{itemize}

			\item \hyperref[samecycle2]{\text{Changed Observation} \ref{samecycle2}} - the same proof works, but \ref{changed_rule:v2_cycle} will imply that $X'$ transfers to $C$ up to $\frac{4}{3}$ points.
			
			\item \hyperref[samecycle3]{\text{Changed Observation} \ref{samecycle3}}-  the same proof works, but \ref{changed_rule:v2_cycle} will imply that $X'$ transfers to $C$ up to $2$ points.
			
			\item \hyperref[samecyclek]{\text{Changed Observation} \ref{samecyclek}} - the same is true and by \hyperref[samecycle2]{\text{Changed Observation} \ref{samecycle2}} and  \hyperref[samecycle3]{\text{Changed Observation} \ref{samecycle3}} we get that  $b(X')\leq |X'|\cdot\frac{2}{3}+2\cdot\frac{2}{3}$.

			\item Changed Case \ref{item3:lemma} in the proof of \hyperref[lem:none_to_dangerous]{\text{Lemma} \ref{lem:none_to_dangerous}} - now $x_i$ has only three free edges and hence can \textbf{go} only to one cycle (instead of two) in addition to $o_1$ and $o_2$. The points $x_i$ will transfer to a cycle is strictly less than $\frac{2}{3}$.   Hence, $b(X)< (k+2)\cdot\frac{2}{3}$.

			\item Changed Case \ref{item4:lemma} in the proof of \hyperref[lem:none_to_dangerous]{\text{Lemma} \ref{lem:none_to_dangerous}} - by \hyperref[samecyclek]{\text{Changed Observation} \ref{samecyclek}} then
			$b(X')\leq|X'|\cdot\frac{2}{3}+2\cdot\frac{2}{3}$, and hence $b(X) = (k+2)\cdot\frac{2}{3}$.
		\end{itemize}
		
		All the rest of the proof remains the same. Hence $X$ ends with at least  $\frac{1}{3}\cdot k-\frac{4}{3}$ points.
	}
	
\end{proofing}

\begin{pro}
	\label{pro:intenal}
	Let $P$ be a path in a canonical path partition. The number of points on internal nodes in $P$ is at least $-1$.
\end{pro}

\begin{proofing}{Proposition \ref{pro:intenal}}
	{
		We show that except for blocks of kind $4$ (the last block) that may end with $-1$ points, the number of points on all blocks is non-negative.
		Because all internal vertices up to the beginning of the first block contribute a non-negative number of points (by item $3$ in \hyperref[pro:easypoints5]{\textbf{Proposition} \ref{pro:easypoints5}}), we have that the internal vertices of a path will have at least $-1$ points.
		
		By item $1$ in \hyperref[pro:easypoints5]{\textbf{Proposition} \ref{pro:easypoints5}} a vertex in $V_2$ has at least $-1$ point , hence blocks of kind $1$ and $2$ are not a problem because they are followed by one or two vertices with at least one point (item $2$ in \hyperref[pro:easypoints5]{\textbf{Proposition} \ref{pro:easypoints5}}). Hence the number of points on any block of kind $1$ or $2$ is non-negative.
		For blocks  of kind $3$,
		\hyperref[pro:lemmafour]{\textbf{Proposition} \ref{pro:lemmafour}} ensures that the number of points on $k\geq 2$ vertices in $V_2^b$ is at least $-1$ (even $-\frac{2}{3}$).

	}
\end{proofing}

Now we prove the theorem:

\begin{proofing}{Theorem \ref{thm:main5}}
	{Consider a canonical path partition $\mathcal{S}$. We show that applying the set of transfer rules to $\mathcal{S}$ leads to a situation where every component in $\mathcal{S}$ has at least $6+\frac{1}{3}$ points. \hyperref[pro:easy6]{\textbf{Proposition} \ref{pro:easy6}} implies that we only need to handle paths. For a path \hyperref[pro:intenal]{\textbf{Proposition} \ref{pro:intenal}} combined with \hyperref[pro:easyouter6]{\textbf{Proposition }\ref{pro:easyouter6}} will imply that the path has at least $6+\frac{1}{3}$ points, which proves the theorem.
		
	}
\end{proofing}

\section*{Acknowledgements}
Work supported in part by the Israel Science Foundation (grant No. 1388/16).

\bibliographystyle{alpha}
\bibliography{proposal_bib}

\begin{thebibliography}{SSSR93}

\bibitem[AEH80]{Akiyama1980}
Jin Akiyama, Geoffrey Exoo, and Frank Harary.
\newblock Covering and packing in graphs. iii: Cyclic and acyclic invariants.
\newblock {\em Mathematica Slovaca}, 30(4):\space 405--417, 1980.

\bibitem[AEH81]{akiyama1981covering}
Jin Akiyama, Geoffrey Exoo, and Frank Harary.
\newblock Covering and packing in graphs iv: Linear arboricity.
\newblock {\em Networks}, 11(1):69--72, 1981.

\bibitem[Alo88]{alon1988linear}
Noga Alon.
\newblock The linear arboricity of graphs.
\newblock {\em Israel Journal of Mathematics}, 62(3):\space 311--325, 1988.

\bibitem[AS04]{alon2004probabilistic}
Noga Alon and Joel~H Spencer.
\newblock {\em The probabilistic method}.
\newblock John Wiley \& Sons, 2004.

\bibitem[CK96]{chang19962}
Gerard~J. Chang and David Kuo.
\newblock The {L}(2,1)-labeling problem on graphs.
\newblock {\em SIAM Journal on Discrete Mathematics}, 9(2):\space 309--316,
  1996.

\bibitem[Dir52]{dirac1952some}
Gabriel~Andrew Dirac.
\newblock Some theorems on abstract graphs.
\newblock {\em Proceedings of the London Mathematical Society}, 3(1):69--81,
  1952.

\bibitem[EP84]{enomoto1984linear}
Hikoe Enomoto and Bernard P{\'e}roche.
\newblock The linear arboricity of some regular graphs.
\newblock {\em Journal of graph theory}, 8(2):309--324, 1984.

\bibitem[FJV19]{ferber2019towards}
Asaf Ferber, Fox Jacob, and Jain Vishesh.
\newblock Towards the linear arboricity conjecture.
\newblock {\em Journal of Combinatorial Theory, Series B}, 2019.

\bibitem[FRS14]{feige2014short}
Uriel Feige, R~Ravi, and Mohit Singh.
\newblock Short tours through large linear forests.
\newblock In {\em International Conference on Integer Programming and
  Combinatorial Optimization}, pages \space 273--284. Springer, 2014.

\bibitem[GJ79]{Hartmanis82}
Michael~R. Garey and David~S. Johnson.
\newblock {\em Computers and Intractability: A Guide to the Theory of
  NP-Completeness}.
\newblock W. H. Freeman \& Co., New York, NY, USA, 1979.

\bibitem[GM60]{gallai1960verallgemeinerung}
Tibor Gallai and Arthur~Norton Milgram.
\newblock Verallgemeinerung eines graphen-theoretischen satzes von {R}edei.
\newblock {\em Acta Scientiarum Mathematicarum}, 21(3-4):\space 181--186, 1960.

\bibitem[Gul86]{guldan1986linear}
Filip Guldan.
\newblock The linear arboricity of $10 $-regular graphs.
\newblock {\em Mathematica Slovaca}, 36(3):\space 225--228, 1986.

\bibitem[Har70]{harary1970covering}
Frank Harary.
\newblock Covering and packing in graphs, i.
\newblock {\em Annals of the New York Academy of Sciences}, 175(1):198--205,
  1970.

\bibitem[HC11]{HUNG2011648}
Ruo-Wei Hung and Maw-Shang Chang.
\newblock Linear-time certifying algorithms for the path cover and
  {H}amiltonian cycle problems on interval graphs.
\newblock {\em Applied Mathematics Letters}, 24(5):\space 648--652, 2011.

\bibitem[Man18]{manuel2018revisiting}
Paul Manuel.
\newblock Revisiting path-type covering and partitioning problems.
\newblock {\em arXiv preprint arXiv:1807.10613}, 2018.

\bibitem[MM09]{MagnantM09}
Colton Magnant and M.~Daniel Martin.
\newblock A note on the path cover number of regular graphs.
\newblock {\em Australasian J. Combinatorics}, 43:\space 211--218, 2009.

\bibitem[Noo75]{noorvash1975}
Shahbaz Noorvash.
\newblock Covering the vertices of a graph by vertex-disjoint paths.
\newblock {\em Pacific Journal of Mathematics}, 58(1):\space 159--168, 1975.

\bibitem[Ore61]{ore1961arc}
Oystein Ore.
\newblock Arc coverings of graphs.
\newblock {\em Annali di Matematica Pura ed Applicata}, 55(1):\space 315--321,
  1961.

\bibitem[PK99]{pak1999optimal}
Wong Pak-Ken.
\newblock Optimal path cover problem on block graphs.
\newblock {\em Theoretical Computer Science}, 225(1-2):\space 163--169, 1999.

\bibitem[Ree96]{Reed96}
Bruce Reed.
\newblock Paths, {S}tars and the {N}umber {T}hree.
\newblock {\em Combinatorics, Probability and Computing}, 5(3):\space
  277–295, 1996.

\bibitem[RW94]{robinson1994almost}
Robert~W. Robinson and Nicholas~C. Wormald.
\newblock Almost all regular graphs are {H}amiltonian.
\newblock {\em Random Structures \& Algorithms}, 5(2):\space 363--374, 1994.

\bibitem[Sku74]{skupien1974path}
Zdzistaw Skupien.
\newblock Path partitions of vertices and {H}amiltonicity of graphs.
\newblock In {\em Proceedings of the Second Czechoslovakian Symposium on Graph
  Theory, Prague}, 1974.

\bibitem[SSSR93]{srikant1993optimal}
R~Srikant, Ravi Sundaram, Karan~Sher Singh, and C~Pandu Rangan.
\newblock Optimal path cover problem on block graphs and bipartite permutation
  graphs.
\newblock {\em Theoretical Computer Science}, 115(2):\space 351--357, 1993.

\bibitem[SW10]{suil2010balloons}
O~Suil and Douglas~B West.
\newblock Balloons, cut-edges, matchings, and total domination in regular
  graphs of odd degree.
\newblock {\em Journal of Graph Theory}, 64(2):\space 116--131, 2010.

\bibitem[SW11]{suil2011matching}
O~Suil and Douglas~B West.
\newblock Matching and edge-connectivity in regular graphs.
\newblock {\em European Journal of Combinatorics}, 32(2):324--329, 2011.

\bibitem[Yu18]{yu2018covering}
Gexin Yu.
\newblock Covering 2-connected 3-regular graphs with disjoint paths.
\newblock {\em Journal of Graph Theory}, 88(3):\space 385--401, 2018.

\end{thebibliography}

\end{document}